\documentclass[11pt]{article}
\usepackage{odonnell}
\usepackage{stmaryrd}
\usepackage{qcircuit}
\usepackage{physics}
 \usepackage{authblk} 
\usepackage{mathtools}

\newcommand{\mc}{\mathcal}
\DeclareMathOperator\eye{\mathbb{I}}
\usepackage[linesnumbered,ruled,vlined]{algorithm2e}
\usepackage[backend=bibtex8,doi=false,isbn=false,url=false,maxbibnames=5,citestyle=numeric-comp,sorting=none,
sortcites=true,
style=numeric-comp]{biblatex}
\bibliography{joels_refs}
\definecolor{darkred}  {rgb}{0.5,0,0}
\definecolor{darkblue} {rgb}{0,0.25,0.7}
\definecolor{darkgreen}{rgb}{0.35,0.6,0.1}
\definecolor{LightPink}{rgb}{0.858, 0.188, 0.478}
\hypersetup{
  urlcolor   = blue,         
  linkcolor  = darkblue,     
  citecolor  = LightPink,    
  filecolor   = darkred       
}
\DeclareMathOperator{\Supp}{Supp_{cg}}

\usepackage{tikz}
\usetikzlibrary{decorations.pathreplacing}
\pgfmathsetmacro\MathAxis{height("$\vcenter{}$")}

\begin{document}

\title{Limitations of Noisy Geometrically Local Quantum Circuits}

\author[1]{Jon Nelson$^*$}
\author[1,2]{Joel Rajakumar$^*$}
\author[1]{Michael J. Gullans}
\affil[1]{\normalsize  Joint Center for Quantum Information \& Computer Science, University of Maryland and NIST}
\affil{\normalsize Department of Computer Science,
	University of Maryland}
\affil[2]{\normalsize  IBM T. J. Watson Research Center, Yorktown Heights, NY}
\date{}
\maketitle
 \def\thefootnote{*}\footnotetext{These authors contributed equally to this work}\def\thefootnote{\arabic{footnote}}
\begin{abstract}
Quantum circuits with a constant rate of depolarizing noise per qubit per time step are known to converge to the uniform distribution at depth $\omega(p^{-1}\log n)$, and hence become trivially classically simulable by uniform sampling. We show that under the physically natural constraint of geometric locality, noisy circuits become classically simulable at shallower depths by substantially more structured classical algorithms. We consider arbitrary geometrically local quantum circuits on $n$ qubits, initialized in an arbitrary product state, with nearest-neighbor gates in $O(1)$ spatial dimensions and interspersed depolarizing noise of any constant strength $p$.

Our first result is that when the depth exceeds a threshold $d^*=\Theta(p^{-1}\log n)$, the output distribution can be approximately sampled in quasipolynomial time. This gives a worst-case simulability result for noisy geometrically local circuits, and matches in its $n$-dependence the best previously known results for noisy random circuits. The proof is based on new information-theoretic bounds showing that in geometrically local noisy circuits, local regions lose correlations and independently converge to maximally mixed strictly before the entire system converges to the uniform distribution.

We further prove structural results suggesting a sharper transition at depth $\tilde{\Theta}(p^{-1})$: after coarse-graining the lattice, Pauli weight supported on long connected paths can be truncated with exponentially small error. These results provide evidence for a percolation-like mechanism behind classical simulability at constant depth, and motivate our conjecture that all noisy geometrically local circuits admit quasipolynomial-time approximate sampling once circuit depth exceeds $\tilde{\Theta}(p^{-1})$.
\end{abstract}

\section{Introduction}

Noise is the central obstacle to large-scale quantum computation. On one hand, the threshold theorem shows that sufficiently low physical noise can be overcome by fault-tolerant protocols, allowing arbitrarily long computations with only moderate overhead \cite{briegel_mbqc,aharonov_ben-or}. On the other hand, near-term and early fault-tolerant devices often operate in a regime where intermediate measurement and reset are limited or expensive, so the computations that are most naturally implemented are predominantly unitary circuits with low but nonzero physical noise \cite{bluvstein_logical,google_phase_transitions,quantinuum_arbitrary_geometries,ibm_evidence_for_utility}. This raises a basic question: how much computational power remains in noisy quantum circuits when active error-correction primitives are unavailable?

A natural starting point is the model of unitary circuits interspersed with single-qubit depolarizing noise. In this setting, a classic result of Aharonov, Ben-Or, Impagliazzo, and Nisan \cite{aharonov_limitations} shows that the output of any noisy circuit on $n$ qubits converges to the uniform distribution by depth $\omega(p^{-1}\log n)$.\footnote{The original paper is formulated for decision problems, but the same convergence mechanism also implies classical easiness for approximate sampling; see, for example, \cite{muller-hermes_stilck-franca_wolf_relative_entropy_convergence,schuster_polynomial,anshu_liu_nguyen_pattison_loglogn_iqp}.} Beyond this depth, approximate sampling is trivial: one can simply sample uniformly. At the same time, that same work shows that noisy circuits can still carry out nontrivial log-depth computation below this threshold. Since logarithmic-depth quantum circuits already capture tasks believed to be classically hard---including factoring \cite{cleve_watrous}, random circuit sampling \cite{bouland_fefferman_nirkhe_vazirani_complexity_rcs}, and IQP sampling \cite{hangleiter_anticoncentration}---this leaves open a broad regime in which noisy \emph{unitary} circuits may still exhibit quantum advantage.

A crucial caveat, however, is that the known positive results for noisy log-depth computation rely on highly nonlocal connectivity. In particular, qubits that are far apart in space can interact as quickly as nearby qubits, effectively bypassing the geometric constraints that dominate large-scale architectures. At finite size some devices can partially approximate this behavior, but at scale all architectures are constrained by space: long-range routing incurs delay, idling, and additional noise. For classical hardware this issue is often hidden by the practical availability of nearly noiseless communication. For quantum hardware it is not. This motivates the following question:
\begin{quote}
\emph{How much computational power do noisy quantum devices retain if they are restricted to nearest-neighbor unitary operations in $O(1)$ spatial dimensions and subject to a constant rate of depolarizing noise?}
\end{quote}

This paper studies that question for the task of approximate sampling from the output distribution. We focus on sampling because, in noisy geometrically local circuits, expectation values of observables can often be classically simulated via lightcone considerations (see \Cref{app:observable}); the possibility of exponential advantage therefore lies primarily in global sampling tasks. Our goal is to understand, as a function of depth, when such circuits can and cannot support quantum sampling advantage.

Our main result shows that geometric locality makes noisy circuits classically simulable \emph{before} they become globally close to uniform. Concretely, for any geometrically local quantum circuit with arbitrary gates, arbitrary product-state input, and constant-rate depolarizing noise, we prove that once the depth exceeds $d^*=\Theta(p^{-1}\log n)$, its output distribution can be approximately sampled in quasipolynomial time. This improves on the trivial uniform-sampling threshold in the geometrically local setting, and gives a worst-case simulability result for universal noisy circuits without imposing randomness, anticoncentration, special gate sets, or large noise.

The key reason this improvement is possible is that in geometrically local circuits, decoherence acts faster than information can be redistributed across space.  We formalize this intuition through new information-theoretic bounds tailored to local noisy dynamics. These bounds have consequences beyond the quasipolynomial-time algorithm above. They imply that after coarse-graining the lattice into blocks of size $\poly(d)$, Pauli weight supported on many blocks is strongly suppressed, and at depth $\tilde{\Theta}(p^{-1})$ one obtains a sharper truncation phenomenon along long connected paths of the coarse-grained lattice. This is precisely the type of structure that appears in percolation-based simulations of noisy geometrically local Clifford circuits \cite{nelson_rajakumar_naturally}, and it suggests that constant-depth simulability may in fact be a universal phenomenon driven by locality rather than by gate-set restrictions.

Motivated by this picture, we formulate a conjectural extension of our results: that every noisy geometrically local circuit becomes classically simulable in quasipolynomial time once its depth exceeds $\tilde{\Theta}(p^{-1})$. We provide supporting structural evidence together with a candidate sampling algorithm whose efficiency would follow from an approximate Markov property of the output distribution. Existing work conjectures that sampling hardness persists in noisy geometrically local circuits at depths below a $\tilde{\Theta}(p^{-1})$ depth threshold \cite{fujii_tamate,rajakumar_watson}, which complements our claims and completes the overall picture.

Aside from quantum computing, a more basic physical motivation for our work is to probe how quantum effects fade at macroscopic scales. Noisy geometrically local quantum circuits are a natural model for physical processes one might encounter in the world -- local depolarizing noise is a simple form of decoherence, physical dynamics are often unitary (reversible), and geometric locality reflects the finite speed at which interactions can spread through space. The existing bound of $\omega(p^{-1}\log n)$ leaves open the possibility that sufficiently large systems could sustain nontrivial long-range behavior for parametrically long times. Our results point in the opposite direction for geometrically local systems: locality and decoherence together appear to force a much earlier loss of long-range structure, and possibly a transition at a depth depending only on the noise strength and not on the system size.

We now summarize our results, compare them with prior work, outline the proof strategy, and list some practical implications.

\subsection{Overview of Results} \label{sec:overview}

The high-level mechanism behind our results is illustrated in \Cref{fig:intuition}. 

\begin{figure}[htbp]
    \centering
    \begin{tabular}{cccc}
        \includegraphics[height=0.13\textheight]{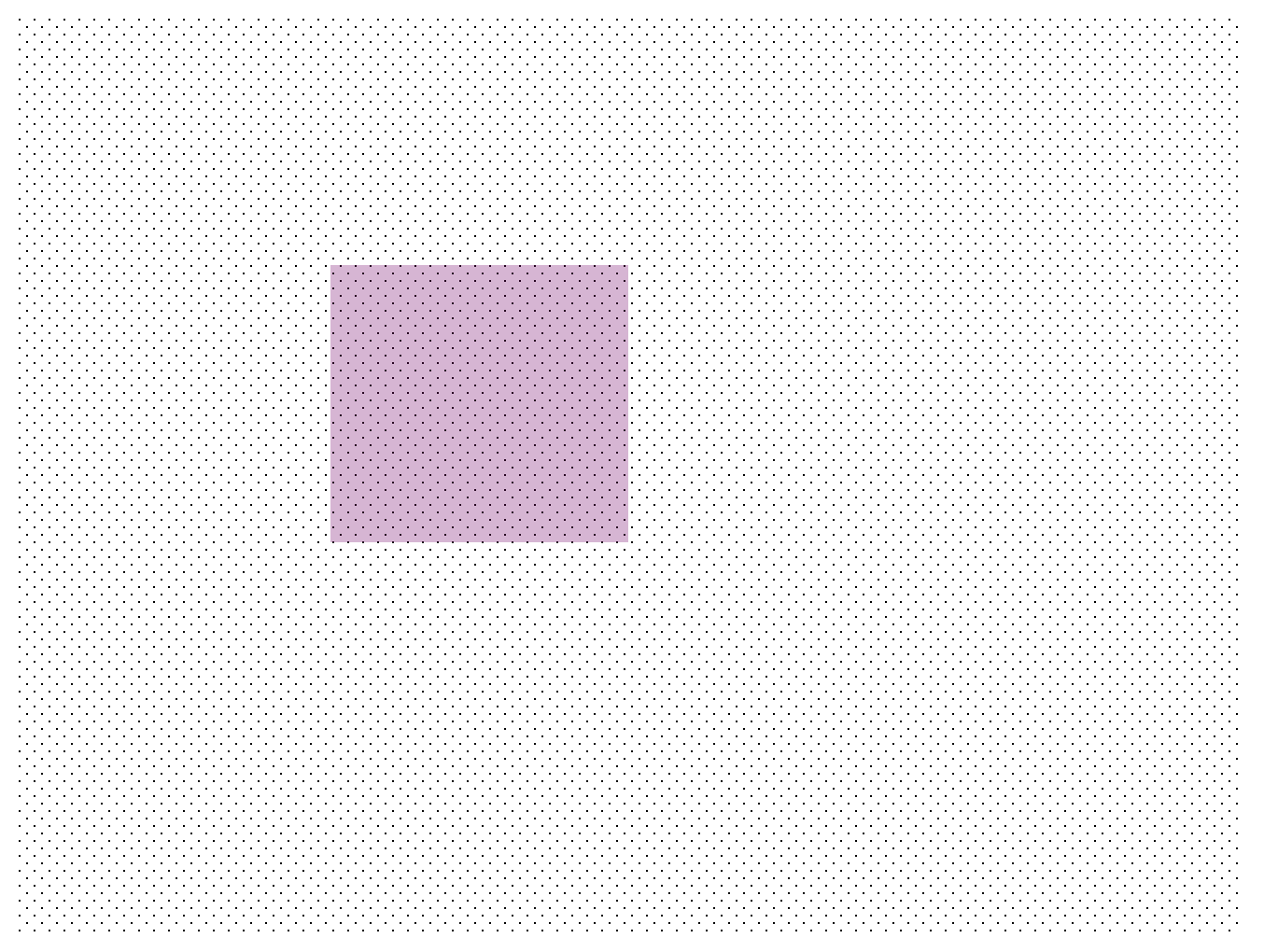} &
        \includegraphics[height=0.13\textheight]{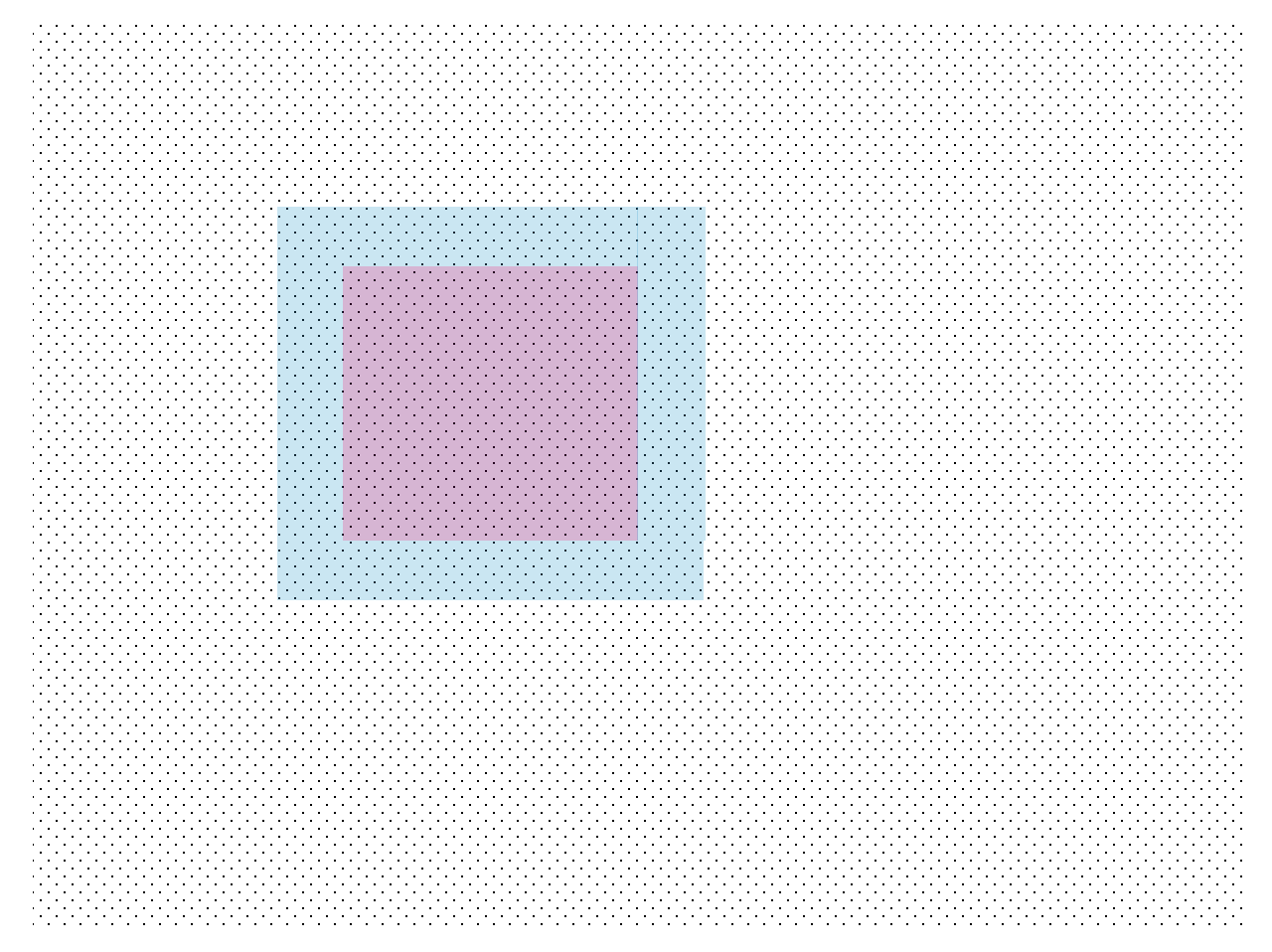} &
        \includegraphics[height=0.13\textheight]{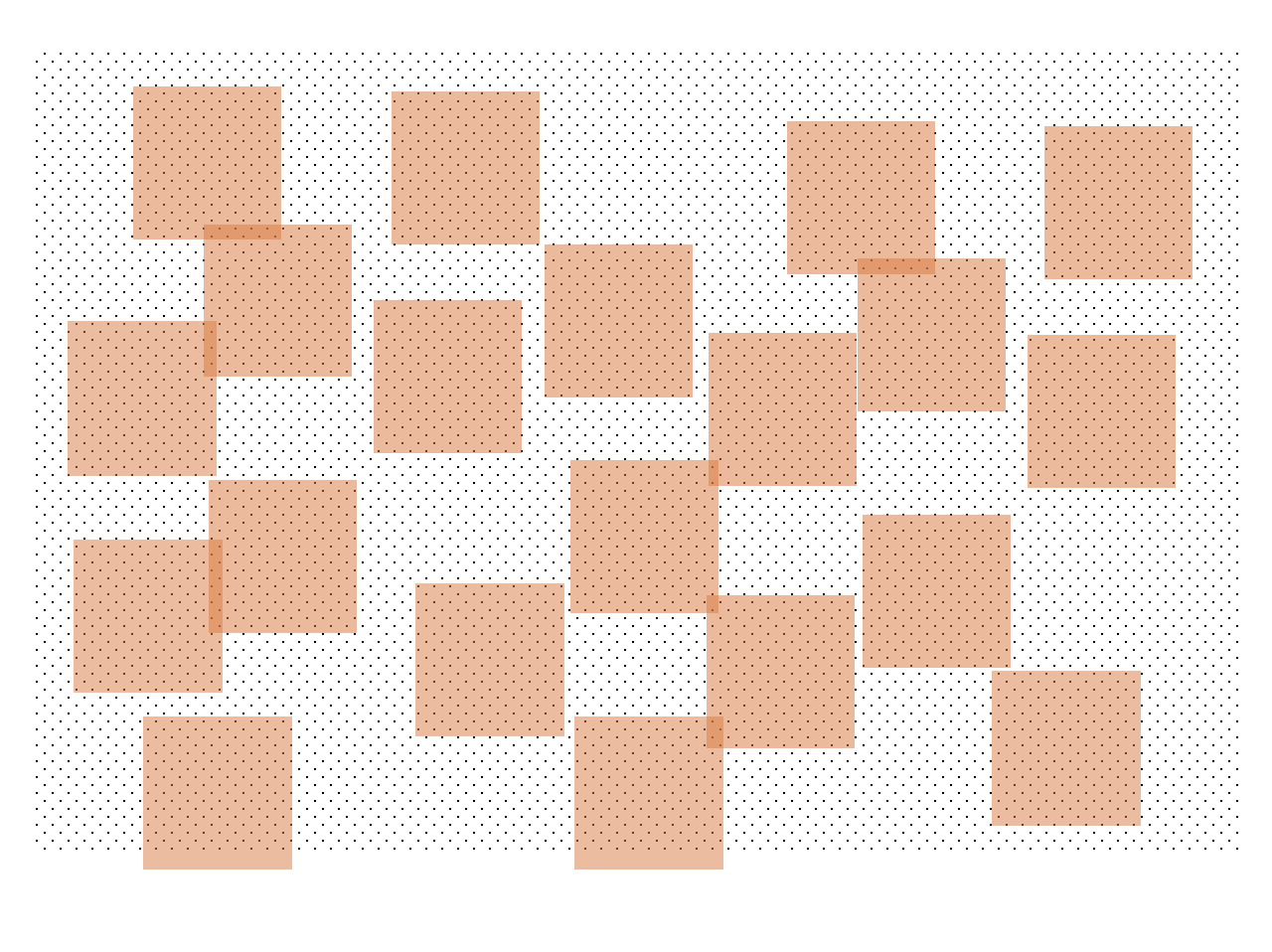} & 
        \includegraphics[height=0.13\textheight]{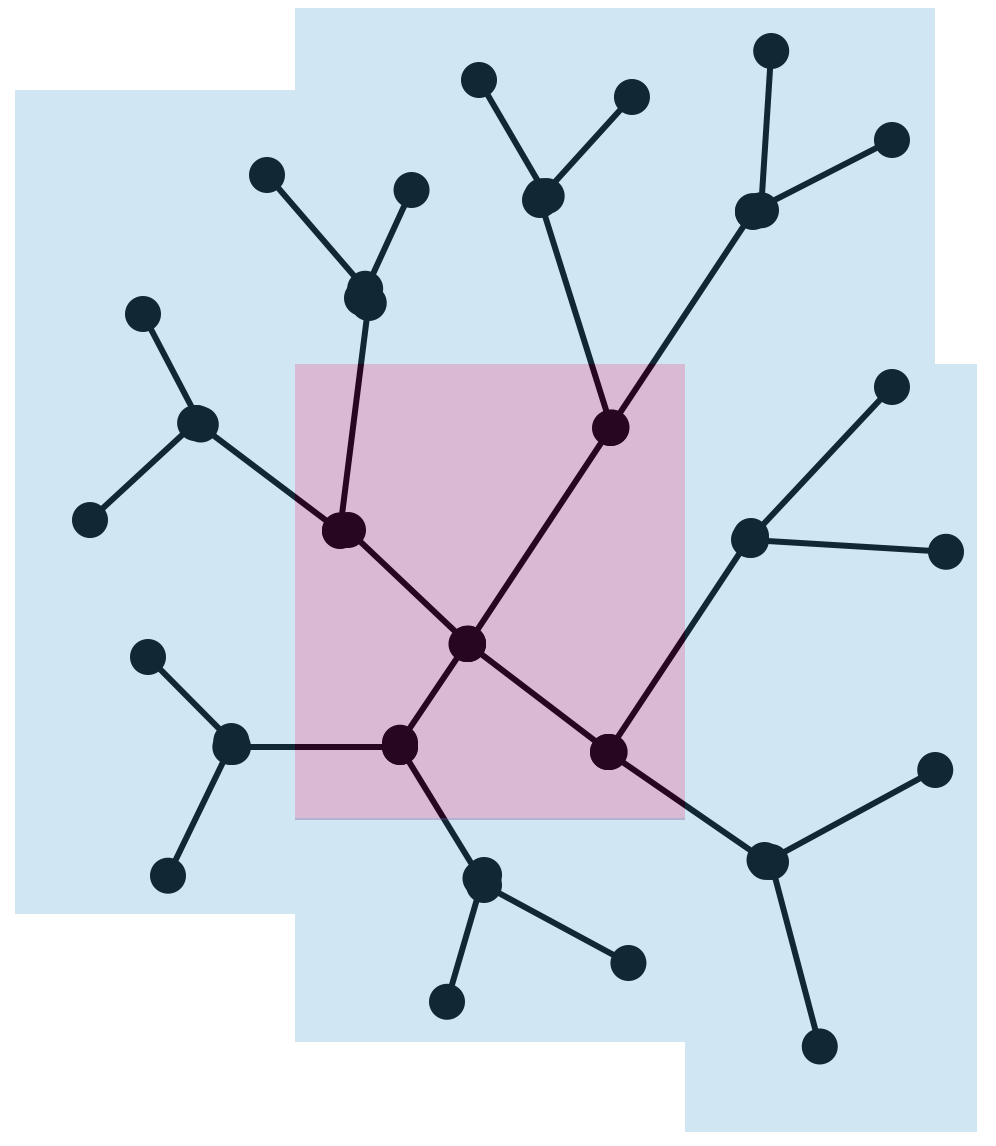} \\
        (a) & (b) & (c) & (d)
    \end{tabular}
    \caption{Consider a noisy geometrically local circuit on a lattice. (a) At each time step, independent noise channels inject entropy throughout the interior of a region. (b) Because the circuit is unitary, gates can only rearrange entropy rather than remove it, and geometric locality restricts this rearrangement to nearest-neighbor transport across the boundary. (c) Since bulk grows faster than boundary, entropy accumulates locally. Small enough regions therefore become close to maximally mixed, which suppresses long-range entanglement and correlations. (d) This mechanism fails on architectures with strong expansion, where boundary can compete with bulk. This is the feature exploited by all-to-all constructions for noisy log-depth computation \cite{aharonov_limitations}}
    \label{fig:intuition}
\end{figure}

We now summarize the three main consequences of this picture.

\subsubsection{Quasipolynomial-Time Approximate Sampling Above the $\Theta(p^{-1}\log n)$ Threshold}

We study approximate sampling from the output distribution of an arbitrary noisy geometrically local quantum circuit. Here ``classically simulable'' means that for any fixed polynomial number of samples, there is an efficient classical algorithm whose output distribution is information-theoretically close to that of the quantum circuit.\footnote{See \cite{aharonov_polynomial_2023} for one motivation for this formulation.} Existing convergence bounds imply that the trace distance to the uniform distribution is at most $e^{-\Omega(pd)}\sqrt{n}$ \cite{aharonov_limitations,muller-hermes_stilck-franca_wolf_relative_entropy_convergence,mele_noise-induced,stilck-franca_garcia-patron}, and hence uniform sampling succeeds once $d=\omega(p^{-1}\log n)$. However, below that depth the output can still be distinguishable from uniform by polynomially many samples, leaving open the possibility of quantum advantage.

Our first result shows that in the geometrically local setting, one can do better than uniform sampling:
\begin{theorem}[Informal]
Let $\Phi$ be any geometrically local quantum circuit acting on $n$ qubits, initialized in an arbitrary product state, with interspersed depolarizing noise of strength $p$ and depth $d$. If $d>d^*$, where $d^*=\Theta(p^{-1}\log n)$, then there exists a quasipolynomial-time classical algorithm that approximately samples from the output distribution of $\Phi$.
\end{theorem}

This gives a worst-case simulability result for universal noisy geometrically local circuits, and matches in its dependence on $n$ the best known simulability results that were previously available only for noisy random circuits \cite{aharonov_polynomial_2023}.

\subsubsection{A Candidate Route to Constant-Depth Simulability}

A more ambitious possibility is that classical simulability already begins at depth $\tilde{\Theta}(p^{-1})$, independent of $n$. This is known for certain non-universal models: for example, noisy geometrically local Clifford circuits with magic-state inputs admit efficient classical simulation through a percolation mechanism \cite{nelson_rajakumar_naturally}. There, interspersed noise effectively cuts the system into small connected components that can be simulated independently.

Our work suggests that this phenomenon may be far more general. We give a patching-type sampling algorithm that does not intrinsically rely on Clifford structure and that is efficient whenever the output distribution satisfies an approximate Markov condition.\footnote{Related patching-based algorithms have also appeared in work on random quantum circuits and approximate Markov structure; see \cite{napp,watts_gosset_liu_soleimanifar,yang_soleimanifar_bergamaschi_preskill,lee_ghosh_approx_markov,zhang_lee_approx_markov}.} This leads to the following conjecture:
\begin{conjecture}[Informal] \label{conj:informal}
For any geometrically local quantum circuit on $n$ qubits with interspersed depolarizing noise of strength $p$ and depth $d$, if $d>d^*$ where $d^*=\tilde{\Theta}(p^{-1})$, then there exists a quasipolynomial-time classical algorithm that approximately samples from its output distribution.
\end{conjecture}

We emphasize that this is a conjectural extension, not a theorem. Our evidence for it comes from the structural results below.

\subsubsection{A Structural Phase Transition at Constant Critical Depth} \label{sec:overview-properties}

The main technical contribution of this paper is a set of information-theoretic bounds that apply to \emph{all} noisy geometrically local circuits. Informally, once the lattice is coarse-grained into blocks of size $\poly(d)$, Pauli operators with support on large coarse-grained structures contribute exponentially little. 
\begin{theorem} [Informal] \label{theorem:evidence}
    Let $\rho$ be the output state of any geometrically local quantum circuit acting on $n$ qubits initialized to any product state, with interspersed depolarizing noise of strength $p$ and depth $d$. Coarse-grain any region of the lattice into `sublattices' of some fixed size $\poly(d)$, labelled by the set $S$. When $d$ exceeds a critical $d^* = \tilde{\Theta}(p^{-1})$ threshold, we can make the following statements about the decomposition of $\rho$ in the Pauli basis (i.e. $\rho$ is a unique sum of Pauli operators),  
    \begin{enumerate}
        \item All Pauli operators with non-identity support (i.e. $X$, $Y$, or $Z$) on every sublattice in $S$ can be truncated, with trace-distance error bounded by $\sim e^{-\Omega(p d |S|)}$.
        \item All Pauli operators with non-identity support on $\ell$ sublattices forming a path in $S$ can be truncated, with trace-distance error bounded by $\sim e^{-\Omega(p\ell d)}|S|$.
    \end{enumerate}
\end{theorem}

The first statement captures the exponential accumulation of noise in both depth and number of sites; the second captures the percolation-type suppression of long connected structures. These are necessary conditions for the constant-depth percolation phenomenon witnessed in \cite{nelson_rajakumar_naturally} and they provide the main evidence for \Cref{conj:informal}. They also have consequences beyond simulation. In particular, we obtain no-go results for single-shot state preparation above depth $\tilde{\Theta}(p^{-1})$, ruling out, for example, the $\ket{0}$-state preparation mechanism used in the fault tolerance protocol of \cite{aharonov_limitations}.

Finally, these bounds are close to tight. At depths below a different threshold $d^{**}=\Theta(p^{-1})$, existing constructions based on noise-robust cluster states and single-shot state preparation \cite{raussendorf_bravyi_harrington_cluster_states,bergamaschi_liu_single_shot,fujii_tamate} violate the structural conclusions above and also violate approximate Markov behavior. Thus our results point to a genuine structural transition around depth $\tilde{\Theta}(p^{-1})$.

\subsection{Comparison to Prior Work} 
\label{sec:prior}

We place our work in the broader landscape of classical simulation results for noisy quantum circuits; see also \Cref{fig:diagram}. Much of the literature studies the easier task of estimating expectation values of fixed observables to inverse-polynomial accuracy \cite{schuster_polynomial,garcia_cirac_trivedi_pauli,fontana,angrisani_mele_simulating,mele_noise-induced,martinez_simulation,angrisani_schmidhuber_rudolph_cerezo_holmes_huang,bravyi_gosset_liu_peaked,bravyi_gosset_movassagh,coble_coudron,dontha_tan_coudron_approximating,stilck-franca_garcia-patron,de-palma_marvian_rouze_stilck-franca_optimal_transport}. In the geometrically local setting, this task is already much easier than approximate sampling, since Pauli observables and small marginals can be estimated directly from lightcone bounds (we make this explicit in \Cref{app:observable}). The main issue here is therefore the complexity of \emph{sampling}.

\begin{figure}[H]
    \centering
    \begin{tikzpicture}

\def\dy{0.75}  
\def\yA{4.5*\dy}   
\def\yB{3.5*\dy}   
\def\yC{2.5*\dy}   
\def\yD{1.5*\dy} 
\def\yE{0.5*\dy}
\def\yaxis{-1*\dy} 





\node[left] at (0,\yaxis) {Depth ($d$)};

\def\dx{1.33}

\draw[<-] (0*\dx,\yA) -- (2*\dx,\yA) node[midway,above] {\cite{fujii_tamate}, hard, exact sampling};
\fill (2*\dx,\yA) circle (2pt);

\draw[dotted] (2*\dx,\yA) -- (8*\dx,\yA) node[midway,above] {?};

\draw[->] (8*\dx,\yA) -- (10*\dx,\yA) node[midway,above] {\cite{aharonov_limitations}, easy};
\fill (8*\dx,\yA) circle (2pt);

\draw[darkblue] (2*\dx,\yB) -- (5*\dx,\yB) node[midway,above] {\textcolor{black}{\Cref{conj:markov}, easy}};
\fill[darkblue] (2*\dx,\yB) circle (2pt);
\fill[darkblue] (5*\dx,\yB) circle (2pt);

\draw[darkblue] (5*\dx,\yB) -- (8*\dx,\yB) node[midway,above] {\textcolor{black}{\Cref{theorem:sampling}, easy}};
\fill[darkblue] (8*\dx,\yB) circle (2pt);

\draw[LightPink] (5*\dx,\yC) -- (8*\dx,\yC) node[midway, above] {\textcolor{black}{\cite{aharonov_polynomial_2023}, easy, random circuits}};
\fill[LightPink]  (5*\dx,\yC) circle (2pt);
\fill[LightPink]  (8*\dx,\yC) circle (2pt);

\draw[LightPink] (2*\dx,\yD) -- (8*\dx,\yD) node[midway,above] {\textcolor{black}{\cite{nelson_rajakumar_naturally}, easy, geo. local Clifford circuits with magic inputs}};
\fill[LightPink]  (2*\dx,\yD) circle (2pt);
\fill[LightPink] (8*\dx,\yD) circle (2pt);

\draw[LightPink][<-] (0*\dx,\yE) -- (8*\dx,\yE) node[midway,above] {\textcolor{black}{\cite{kempe_upper_bounds}, easy, $p > 0.357$}};
\fill[LightPink] (8*\dx,\yE) circle (2pt);

\draw[->] (0*\dx,\yaxis) -- (10*\dx,\yaxis);

\foreach \x/\lbl in {
  2*\dx/{$\Tilde{\Theta}(p^{-1})$},
  5*\dx/{$\Tilde{\Theta}(\log n)$},
  8*\dx/{$\omega(p^{-1} \log n)$}
}{
  \draw (\x,\yaxis+0.1) -- (\x,\yaxis-0.1);
  \node[below=6pt] at (\x,\yaxis) {\lbl};
}

\end{tikzpicture}
    \caption{Sampling from the output distribution of noisy geometrically local quantum circuits on $n$ qubits, with depth $d$ and depolarizing noise strength $p$. Black indicates previously established regimes, blue indicates our contributions, and pink highlights selected prior results that require additional assumptions such as randomness, non-universal gate sets, or larger noise.}
    \label{fig:diagram}
\end{figure}

For approximating output distributions, essentially all prior and concurrent positive results in related settings rely on at least one additional assumption: randomness \cite{aharonov_polynomial_2023,bremner_achieving,gao_duan,takahashi_ct-ecs,schuster_polynomial,lee_ghosh_approx_markov}, non-universal gate sets \cite{nelson_rajakumar_naturally,rajakumar_watson_liu,oh_classical_simulability_linear_optical}, or noise above a fixed threshold \cite{aharonov_polynomial_1996,aharonov_quantum_to_classical,fujii_tamate,cheng_ippoliti,seddon_magic_monotones,kempe_upper_bounds,plenio_upper_bounds,burhman_upper_bounds,razborov_upper_bounds,zhang_lee_approx_markov}. Our results differ in that we treat \emph{worst-case} circuits with a \emph{universal} gate set and \emph{arbitrarily low but constant} depolarizing noise, assuming only geometric locality. We highlight three particularly relevant points of comparison.

First, if one relaxes to noisy \emph{average-case} circuits, then the best known classical algorithms also succeed only once the depth is $\Theta(p^{-1}\log n)$ \cite{aharonov_polynomial_2023}. Our first theorem matches this scaling in $n$, but in a much more adversarial setting: we do not assume randomness, anticoncentration, or average-case behavior. The reason this is possible is that we exploit geometric locality in place of randomness.

Second, as mentioned earlier, classical simulability has been proven for noisy geometrically local Clifford-magic and IQP+CNOT circuits with depth that exceed $\tilde{\Theta}(p^{-1})$ \cite{nelson_rajakumar_naturally}. Those analyses rely on special noise propagation properties, whereas our proofs apply to universal circuits and identify geometric locality as the true underlying mechanism behind noise accumulation and percolation.

Third, there is a distinct line of work showing simulability once the noise rate exceeds a fixed constant threshold. One might be tempted to read our results as a statement of that kind for each fixed depth $d$: namely, that there exists some threshold $p^*$ above which simulation is easy. But fixed-threshold results of this kind are already known, for example at $p^*=0.357$ \cite{kempe_upper_bounds}. Our result is stronger in a different direction: the critical noise level decreases roughly inversely with the circuit depth, reflecting the fact that noise accumulates over time faster than local gates can distribute redundancy.

\subsection{Overview of Proof Techniques} \label{sec:proof}

Our results for worst-case circuits come from a new information-theoretic perspective on noisy local dynamics. Existing convergence arguments show that if $\rho$ is the output state of a noisy circuit and $\sigma$ is the maximally mixed state on all qubits, then
\[
D(\rho\|\sigma)\leq (1-p)^d n,
\]
which explains why the full output distribution becomes close to uniform at depth $\omega(p^{-1}\log n)$ \cite{aharonov_limitations,muller-hermes_stilck-franca_wolf_relative_entropy_convergence,aharonov_polynomial_1996,stilck-franca_garcia-patron}. This global statement is too weak for our purposes, because it only becomes useful once the entire system has already mixed.

Our first key step is a localized version of this convergence phenomenon. For any subset of qubits $A$, we prove a bound of the form
\[
D\!\left(\rho \,\middle\|\, \sigma_A\otimes \Tr_A(\rho)\right)\leq (1-p)^d |L(A)|,
\]
where $L(A)$ is the reverse lightcone of $A$. This says that the qubits in $A$ not only approach the maximally mixed state, but also lose correlations with the rest of the system, at a rate controlled by the size of their reverse lightcone. The proof combines existing relative-entropy techniques for depolarizing channels with a conditional quantum Shearer inequality \cite{berta_conditional_shearers}.

This localized bound becomes powerful because of geometric locality. In bounded spatial dimension, reverse lightcones grow only polynomially with depth. Thus if $A$ is chosen to be a sublattice of side length $O(d)$, the quantity $|L(A)|$ is polynomial in $d$, while the prefactor $(1-p)^d$ decays exponentially in $d$. It follows that individual sublattices become close to maximally mixed after depth only $\tilde{\Theta}(p^{-1})$, well before global convergence.

The next step is to show that this local convergence can be exploited simultaneously across many well-separated sublattices. We develop an inclusion--exclusion decomposition of the output state in which each term corresponds, equivalently, to a family of Pauli operators with prescribed coarse-grained support. Terms whose support spans many sublattices acquire exponentially small trace norm. Once this structure is established, two different truncation regimes emerge.

At depth $\Theta(p^{-1}\log n)$, one can truncate every Pauli term supported on more than $O(\log n)$ sublattices while incurring only bounded trace-distance error. The remaining terms can be enumerated by brute force, their marginals can be computed efficiently, and a standard sampling-to-computing reduction \cite{bremner_achieving} yields quasipolynomial-time approximate sampling.

At depth $\tilde{\Theta}(p^{-1})$, one can truncate all Pauli terms supported on large \emph{connected} components of the coarse-grained lattice. This is proved through a site-percolation argument reminiscent of earlier simulations of non-universal noisy models at constant depth \cite{rajakumar_watson_liu,nelson_rajakumar_naturally,oh_classical_simulability_linear_optical}. The resulting structural theorem is the main evidence for our constant-depth conjecture.

Part of the overall strategy is reminiscent of the Pauli-path framework used in prior classical algorithms \cite{aharonov_polynomial_2023,bremner_achieving,gao_duan,takahashi_ct-ecs,schuster_polynomial,fontana,angrisani_schmidhuber_rudolph_cerezo_holmes_huang,angrisani_mele_simulating,mele_noise-induced,martinez_simulation}. The key difference is that those methods enumerate all possible \emph{paths} of Pauli operators through the circuit, and the resulting combinatorial blowup must be offset by additional assumptions such as randomness. Our approach avoids that blowup entirely: rather than track every path separately, we use information-theoretic arguments to bound the \emph{total} trace norm of all paths that contribute to output Pauli operators with large coarse-grained support.

\subsection{Practical Implications of the Main Conjecture} \label{sec:discussion}

If \Cref{conj:informal} is true, it gives a clean architectural interpretation of what resources are needed to sustain quantum advantage in the presence of noise.

\begin{itemize}
    \item \textbf{Near-term devices.} In the absence of active fault tolerance, useful sampling experiments must finish quickly: increasing the system size alone does not buy additional computational time -- only reducing the physical noise rate can.
    \item \textbf{Early fault-tolerant devices.} Nonlocal connectivity plays a role analogous to dissipative primitives such as intermediate measurement and reset. Both are scarce resources that can circumvent the local-heating mechanism underlying our simulation results.
    \item \textbf{Large-scale fault tolerance.} If logical noise resembles depolarizing noise and the logical circuit remains unitary and geometrically local, then sustaining asymptotic quantum advantage may require logical error rates to fall at least inversely with the target circuit depth. Equivalently, the conjectured picture suggests that code sizes must grow at least logarithmically with the intended logical runtime.
\end{itemize}

\section{Preliminaries and Notation} \label{sec:prelims}

In general, we will consider all matrices to be defined on the Hilbert space corresponding to qubits indexed by the set $[n]$. We will often consider subsets of this set, which we label by capital letters, e.g. $A \subseteq [n]$. We use $\mathrm D(A)$ to represent the set of density operators on the Hilbert space corresponding to the qubits in $A$. For disjoint sets $A$ and $B$, we will often use $AB$ as shorthand for the union of $A$ and $B$. For $A,B \subseteq [n]$, we will use $B \backslash A$ to denote $B - A \cap B$. We also use $\overline{A} := [n] \backslash A$. For any matrix $\varrho$ defined on the Hilbert space of $n$ qubits, we use $\varrho_A := \Tr_{\overline{A}}(\varrho)$. We will use $\sigma_A$ and $\ketbra{0}_A$ to denote the density matrix on qubits in the set $A$ corresponding to the maximally mixed state, $\eye/2^{|A|}$, and the all-zero state, $\ketbra{0}^{\otimes |A|}$, respectively.

We will generally use the font $\mc  A, \mc B, \mc C \ldots$ to denote linear maps and/or channels, and we insert subscripts, e.g. $\mc A_i$ to denote the target qubit(s) if applicable. For maps which act non-trivially on only one qubit, we will use $\mc A_{S} =  \bigcirc_{i \in S}\mc A_i$, for any set of qubits $S \subseteq [n]$, to denote a composition of these maps on each qubit in the set. For any qubit $i \in [n]$ and any matrix $\varrho$, we define the identity channel $\mc I_i$, the complete depolarizing channel $\mc D_i$, and the depolarizing channel of fixed strength $p$, $\mc N_i$, as follows,
\begin{align*}
    \mc{I}_i(\varrho) &= \varrho\\
    \mc{D}_i(\varrho) &= \sigma_i \otimes \varrho_{\overline{i}} \\
    \mc{N}_i(\varrho) &= (1-p)\varrho + p \sigma_i \otimes \varrho_{\overline{i}}
\end{align*}

\begin{definition}[The Output State of a Noisy Quantum Circuit] \label{def:circuit} 
    We use $\Phi$ to denote the channel corresponding to a circuit of $d$ layers of two-qubit unitary gates on $n$ qubits, where single-qubit depolarizing noise of strength $p$ is applied on each qubit after each layer, i.e. $\Phi = \bigcirc_{i \in d,\ldots ,1} [\mc N_{[n]} \circ \mc U^{(i)}]$, where $\mc U^{(i)}$ is the $i^{th}$ layer of unitaries in the circuit. We use $\rho$ to refer to the density matrix produced by applying $\Phi$ on an arbitrary input product state (without loss of generality, we can assume this input to be $\ketbra{0}^{\otimes n}$). We will use the phrase `$\rho$ is the output state of a noisy quantum circuit,' to refer to this setup. 
\end{definition}

Note, we will often only be concerned with the output distribution of such a noisy quantum circuit when measured in a product basis (without loss of generality, we can assume this basis to be the computational basis), in which case it can be assumed that $\rho$ is completely dephased on all qubits. However, all our results apply broadly to the output state and the dephased output state. We will use $P$ to generally refer to the output distribution over bitstrings produced by measuring some $\rho$ in the computational basis. For any $A \subseteq [n]$, we will also use $P_A$ to denote the marginal distribution on $A$, and for any bitstring $b$ on qubits $B \subseteq [n]$, we use $P_{A |B = b}$ to denote the distribution on $A$, conditioning on the event that the bistring $b$ is sampled on $B$. 

We will use $\mathsf{P}_{n} = \{I,X,Y,Z\}^{\otimes n}$ denote the set of all Pauli operators on $n$ qubits.
    
    We call a quantum circuit `geometrically local' if its qubits can be placed on a lattice such that every gate in the circuit is nearest-neighbor on this lattice. We will use $D$ to refer to the dimensionality of this lattice (e.g. $D=2$ on a 2D grid). We will assume $D = O(1)$ throughout this work, and thus will omit dependence on $D$ in most cases.

\begin{definition}[Coarse-Graining] \label{def:coarse-grain}
    For geometrically local circuits, we will coarse-grain the lattice into a fixed set of $m$ `sublattices,' each of size $(2d)^D$ which have side lengths $2d \times 2d \ldots \times 2d$. We will denote the set of all sublattices as $J$. We will fix some ordering of the sublattices $J_1,\ldots, J_,$. We will define $J_{\leq i} = \bigcup_{k \leq i} J_k$, and define $J_{< k}$ similarly. For sublattice $A \in J$, we use $\partial^{\ell}A$ to denote the set of sublattices that are at most $\ell$ sublattices away from $A$ (including $A$). We denote $\partial A = \partial^1A$. For any $P \in \mathsf{P}_n$, let $\Supp(P) \subseteq  J$ denote the set of sublattices on which $P$ acts non-trivially. 
\end{definition}

\section{Information-Theoretic Properties of Noisy Geometrically Local Quantum  Circuits} \label{sec:information-theory}

\subsection{Convergence of Relative Entropy on Subsets of Qubits}
Relative entropy convergence under tensor products of depolarizing channels \cite{muller-hermes_stilck-franca_wolf_relative_entropy_convergence} is a powerful tool used to prove many existing results for `worst-case' noisy quantum circuits \cite{stilck-franca_garcia-patron,de-palma_marvian_rouze_stilck-franca_optimal_transport,aharonov_limitations,fawzi_muller-hermes_shayeghi_lower_bound}. In particular, when any state $\rho$ experiences depolarizing channels on \textit{each} qubit, the relative entropy $D(\rho\|\sigma)$, which characterizes the distance to the maximally mixed state, decays by a multiplicative factor. Here, we make a modification to this tool which allows us to handle depolarizing channels on \textit{subsets} of qubits rather than on all qubits. Note the following lemma applies to \textit{any quantum state}, not just those arising from noisy quantum circuits.
\begin{lemma} [Relative Entropy Convergence on Subsets]\label{lemma:relative_entropy_convergence_subsets}
For any state $\rho$ on $n$ qubits and $A \subseteq [n]$,
\begin{align}
        D(\mc N_A(\rho)\|\sigma_A \otimes \rho_{\bar{A}}) \leq (1-p)D(\rho\|\sigma_A \otimes \rho_{\bar{A}})
    \end{align}
\end{lemma}
This is proven in \Cref{app:relative_entropy_convergence_subsets}, and it follows the general proof strategy of an earlier proof of relative entropy convergence on all qubits \cite{muller-hermes_stilck-franca_wolf_relative_entropy_convergence}, applying a `conditional quantum Shearer's inequality' \cite{berta_conditional_shearers} at an intermediate step. 



In existing results, relative entropy convergence is usually used inductively at each layer of the circuit to show that $D(\rho\|\sigma) \leq (1-p)^dn$ for any noisy quantum circuit. Here, we make a modification to this induction technique, showing that for any subset of qubits $A$, the output of the circuit converges exponentially quickly to $\sigma_A \otimes \rho_{\overline{A}}$, but incurring an overhead proportional to the size of the `reverse lightcone' $L(A)$, rather than $n$. This is depicted pictorially in \Cref{fig:2} (we define the lightcone more formally in \Cref{app:exponential_decay}). Note, the following lemma is general to \textit{any noisy quantum circuit}, not just geometrically local ones.

\begin{lemma} [Relative Entropy Decay in Circuits]\label{lemma:exponential_decay}
    Let $\rho$ be the output state of a noisy quantum circuit as in \Cref{def:circuit}. Let $A \subseteq [n]$, and $L(A)$ be the set of qubits in its reverse lightcone.
    \begin{align}
        D(\rho\|\sigma_A \otimes \rho_{\overline{A}}) \leq (1-p)^d|L(A)|
    \end{align}
\end{lemma}
This is proven in \Cref{app:exponential_decay}. 

\begin{figure}[h]
  \centering
  \makebox[\textwidth]{\includegraphics[width=1.0\linewidth]{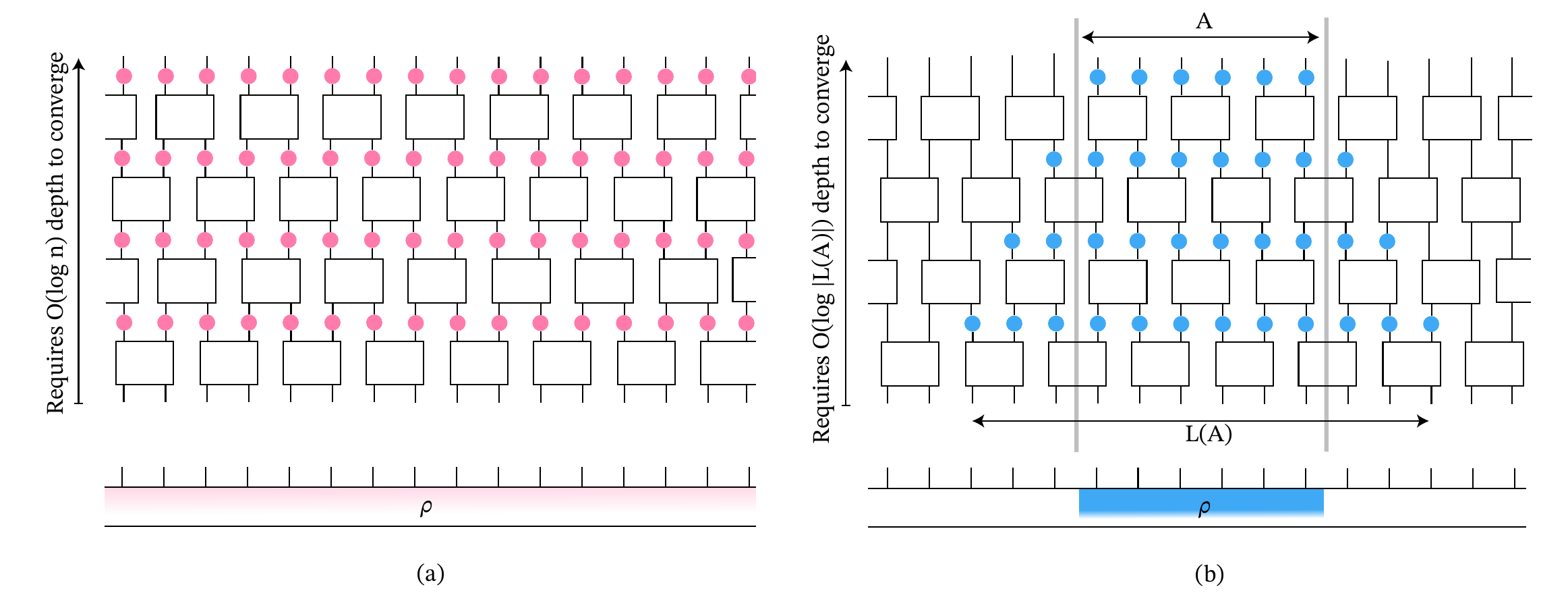}}
  \caption{Figure (a) shows how convergence of the entire state to within a constant relative entropy distance of $\sigma$ requires noise on all qubits, and further $O(\log n)$ layers of this noise. Figure (b) shows how convergence to $\sigma_{A} \otimes \rho_{\overline{A}}$ only requires noise within the lightcone of $A$, and further this means only $O(\log(|L(A)|))$ layers are required.}
  \label{fig:2}
\end{figure}

\subsection{Coarse-Graining into Sublattices and Convergence at Constant Depth}
Under the constraint of geometric locality, lightcones cannot grow very quickly, and this motivates us to consider small `sublattices' in the lattice, which have bounded lightcone size. These are defined in \Cref{def:coarse-grain}. One important feature of this coarse-graining is that it allows different sublattices to `independently' converge, since their lightcones do not intersect. This is depicted in \Cref{fig:3}. We now have the following corollary of \Cref{lemma:exponential_decay}. 
\begin{corollary} [Relative Entropy Decay in Geo. Local Circuits]\label{corollary:exponential_decay}
    When $\rho$ is the output state of a geometrically local noisy quantum circuit as in \Cref{def:circuit}, for any $J_i \in J$, 
    \begin{align}
        D(\rho\|\sigma_{J_i} \otimes \rho_{J \backslash J_i}) \leq (1-p)^d(4d)^D
    \end{align} 
\end{corollary}
Our main insight is that, since $(1-p)^d$ is exponentially decaying in depth while $(4d)^D$ is only polynomially growing in depth, this quantity becomes less than $1$ after some constant critical depth threshold and inverse polynomially small after a log-depth threshold. In particular, we note the following fact,
\begin{fact} [Follows from Lemma 20 of \cite{rajakumar_watson_liu}] \label{fact:critical_depth}
    For any $c >1$, there exists some $d^* = \Theta(p^{-1}\log (p^{-1}c))$, such that when $d > d^*$, $(1-p)^d (4d)^D < 1/c$.
\end{fact}

\begin{figure}[h]
  \centering
  \makebox[\textwidth]{\includegraphics[width=1.0\linewidth]{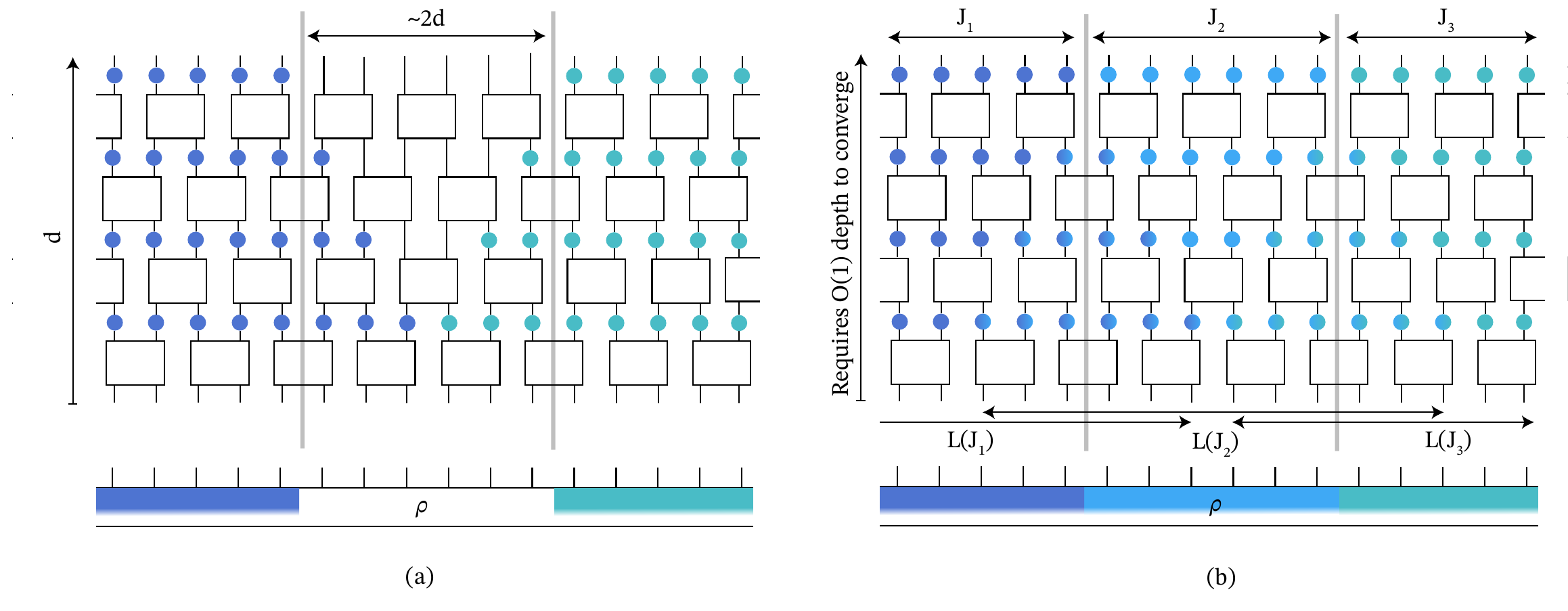}}
  \caption{In Fig (a), we show how non-adjacent sublattices independently converge towards the maximally mixed state due to our specific choice of coarse-graining. Fig (b) depicts what our relative entropy bounds tell us about the state after constant depth. This can be compared with \Cref{fig:2}(a): while $O(\log n)$ depth is required for the full state to be within constant relative entropy distance to the maximally mixed state, we show that after $O(1)$ depth, every single sublattice is within a constant relative entropy distance to the maximally mixed state on that sublattice.}
  \label{fig:3}
\end{figure}

\section{Truncatability of Pauli Operators at Critical Depth} \label{sec:truncatability}
In order to apply \Cref{corollary:exponential_decay}, we would like to decompose $\rho$ into terms that agree with this convergence and error terms that deviate from this convergence. We can then truncate these error terms and bound their contribution to arrive at our approximation. To do this, first consider trivially decomposing $\rho$ as $\sigma_{J_i} \otimes \rho_{J \backslash J_i} + (\rho -\sigma_{J_i} \otimes \rho_{J \backslash J_i})$ for a given sublattice $J_i$. This is equivalent to applying the map $\mc D_{J_i} +(\mc I - \mc D_{J_i})$. The first term is what $\rho$ converges towards, and the second term is considered the error term, and its trace norm can be directly bounded by \Cref{corollary:exponential_decay} through Pinsker's inequality. We can now recursively apply this decomposition for every sublattice, which results in what we call the ``inclusion-exclusion decomposition''. Once $\rho$ is in this form, we show that the trace norm of a term in the decomposition is exponentially decayed by the number of sublattices on which it acts as an error term. It turns out that the inclusion-exclusion decomposition can be equivalently interpreted as a decomposition into sets of Pauli operators with varying support on sublattices. We introduce notation for Pauli operators and their support below,
\begin{definition}[Coarse-grained Support]
     For any $P \in \mathsf{P}_n$, let $\Supp(P) \subseteq  J$ denote the set of sublattices on which $P$ acts non-trivially. 
\end{definition}

\begin{figure}[h]
  \centering
  \makebox[\textwidth]{\includegraphics[width=1.0\linewidth]{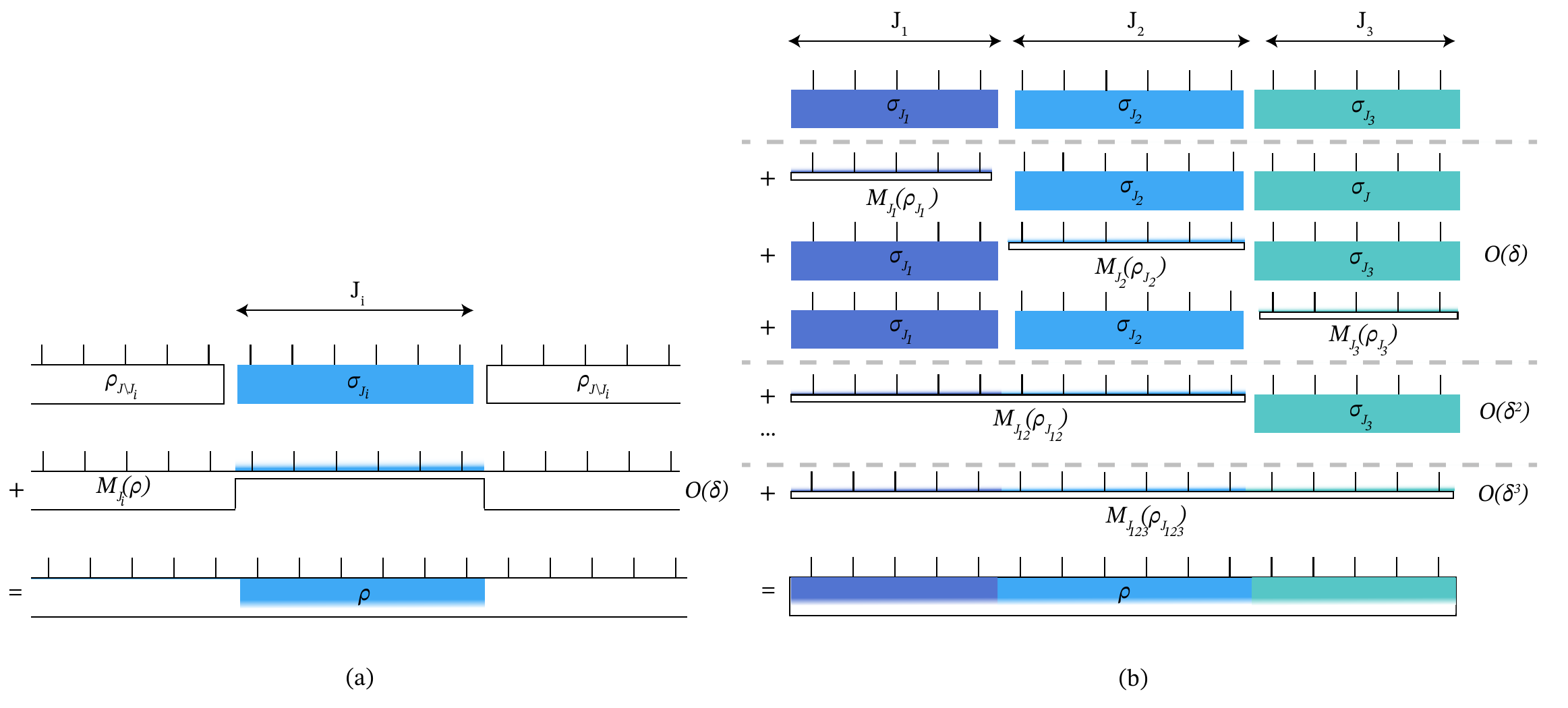}}
  \caption{In Fig. (a), we depict how $\rho$ can be trivially decomposed into $\sigma_{J_i} \otimes \rho_{J \backslash J_i}$ and a residual term $ \rho - \sigma_{J_i} \otimes \rho_{J \backslash J_i} =\mc M_{J_i}(\rho)$. We use a completely blue box to depict $\sigma_{J_i}$ since it is a maximally mixed state. We use a thinner box with shading above it to denote $\mc M_{J_i}(\rho)$, since it is not a state, but rather a trace-zero matrix of small, $O(\delta)$, trace norm. In Fig. (b), we show how this decomposition can be applied recursively for multiple sublattices, and we show how each residual term is exponentially suppressed (note the scaling in $\delta$), which allows them to be truncated without incurring too much error.}
  \label{fig:4}
\end{figure}

\subsection{Inclusion-Exclusion Decomposition and Grouping of Pauli Basis Terms}

\begin{definition} [Inclusion-Exclusion Map] \label{def:inclusion-exclusion}
    For any $J_i \in J$ and any matrix $\varrho$, we define,
    \begin{align*}
        \mc M_{J_i}(\varrho) = \varrho - \sigma_{J_i} \otimes \varrho_{J \backslash J_i} =  [\mc  I - \mc D_{J_i}] (\varrho)
    \end{align*}
\end{definition}
    
Notice that $\varrho = [\mc D_{J_i} + \mc  M_{J_i}](\varrho)$. We can apply this identity map for each sublattice to get the following inclusion-exclusion decomposition,
    \begin{align}
        \varrho &= \bigcirc_{J_i \in J} [\mc D_{J_i} + \mc  M_{J_i}] (\varrho)\\
        &= \sum_{A \subseteq J} [\mc M_{A} \otimes \mc D_{J\backslash A}](\varrho) \\
        & = \sum_{A \subseteq J} \mc M_A(\varrho_A) \otimes \sigma_{J \backslash A} 
    \end{align}
    We call this an inclusion-exclusion decomposition because $\mc M_A$ has the following form,
    \begin{align}
        \mc M_{A}(\varrho) = \bigcirc_{J_i \in A}\mc M_{J_i}(\varrho) = \sum_{B \subseteq A}(-1)^{|B|} \sigma_B \otimes \varrho_{J \backslash B} 
    \end{align}

Our inclusion-exclusion technique essentially provides an analytical tool to isolate out Pauli operators of large support, as described in the following lemma,
\begin{lemma}[Pauli Basis Interpretation of \Cref{def:inclusion-exclusion}]\label{lemma:pauli_basis}
For any $A \subseteq J$ and $P \in \mathsf{P}_n$,
\begin{align}
    \mc M_A (P) = \begin{cases}
            P, &\text{if }A \subseteq \Supp(P)\\
            0, &\text{if }A \not \subseteq \Supp(P)
        \end{cases}
\end{align}
Moreover,
\begin{align}
        \mc M_A\otimes \mc D_{J\backslash A}(P) = \begin{cases}
            P, &\text{if }A = \Supp(P)\\
            0, &\text{if }A \neq \Supp(P)
        \end{cases}
    \end{align}
\end{lemma}

\subsection{Exponential Decay of Inclusion-Exclusion Terms/Sums of Pauli Operators}
By carefully exploiting the independent convergence of spatially separated sublattices due to \Cref{corollary:exponential_decay}, we can bound the trace norm of each term in the inclusion-exclusion decomposition, summarized in the following key technical lemma,
\begin{lemma} [Bounds on Inclusion-Exclusion Terms]\label{lemma:inclusion-exclusion}
    Let $\rho$ be the output state of a noisy geometrically local quantum circuit as in \Cref{def:circuit}, and let $J$ denote the coarse-graining of \Cref{def:coarse-grain}.. For any set of sublattices $A \subseteq J$
        \begin{align}
            \|\mc M_A(\rho)\|_1 = \norm{\sum_{P \in \mathsf{P}_{n}: A \subseteq \Supp(P)} \frac{\Tr(\rho P)}{2^n}P }_1&\leq (2[(1-p)^d(4d)^D]^{1/(2 * 3^D)})^{|A|}
        \end{align} 
    In particular, when $D=O(1)$ and $d > d^*$ where $d^* = \Theta(p^{-1}\log (p^{-1}))$, this bound scales like $e^{-\Omega(|A|pd)}$
\end{lemma}

We remark that \Cref{lemma:inclusion-exclusion} essentially places a bound on the \textit{trace norm of the sum of all Pauli operators} in this decomposition with support that includes $A$, and this bound exponentially decays in the (1) support of the Pauli operators, (2) noise strength, and (3) depth. This is one of the main technical contributions of our work. Previous Pauli based techniques \cite{aharonov_polynomial_2023,bremner_achieving,gao_duan,takahashi_ct-ecs,schuster_polynomial,fontana,angrisani_schmidhuber_rudolph_cerezo_holmes_huang,angrisani_mele_simulating,mele_noise-induced,martinez_simulation} used the fact that \textit{each} Pauli operator exponentially decays in its support throughout its evolution during the circuit, whereas we have managed to achieve the stronger result that many of these terms can be grouped together, and their total trace norm can be bounded.

\subsection{Sparse and Percolated Approximations}
Now we will consider truncating terms in the decomposition of $\rho$ in the inclusion-exclusion decomposition. Recall that $\rho$ can be decomposed into inclusion-exclusion terms as follows, 
\begin{align}
    \rho &= \sum_{A \subseteq J} \mc M_A(\rho_A) \otimes \sigma_{J \backslash A} \label{eq:inclusion-exclusion}
    \end{align}
Or equivalently, in the Pauli basis, 
    \begin{align}  
    \rho &= \sum_{A \subseteq J} \sum_{P \in \mathsf{P}_{n}: \Supp(P) = A} \frac{\Tr(\rho P)}{2^n}P 
\end{align}
First, we consider truncating all terms of \Cref{eq:inclusion-exclusion} with $|A| > k$ for some parameter $k$. We call this a `sparse approximation.' We can apply the exponentially decaying bounds on each inclusion-exclusion term proven in \Cref{lemma:inclusion-exclusion} to bound the error in trace distance incurred by these sparse approximations. In particular, when the depth exceeds a $\Theta(p^{-1}\log (n))$ critical threshold, the exponential decay of each term is so strong (inverse polynomial in $n$ and exponential in $|A|$) that it overwhelms the combinatorial number of terms that have weight greater than $k$. This results in our first theorem,
\begin{theorem} [Convergence to sparse approximations] \label{theorem:convergence1}
    Let $\rho$ be the output state of a noisy geometrically local quantum circuit as in \Cref{def:circuit} with $d >d ^*$, where $d^* = \Theta(p^{-1}\log (n))$, and let $J$ denote the coarse-graining of \Cref{def:coarse-grain}. For any parameter $k \in [m]$, define $\rho_{sparse,k}$ as follows,
    \begin{align}
    \rho_{sparse,k} := \sum_{A \subseteq J: |A| \leq k} \mc M_A(\rho_A) \otimes \sigma_{J \backslash A} \label{eq:sparse_approximation}    \end{align}
Or equivalently, in the Pauli basis, 
    \begin{align}  
    \rho_{sparse,k} := \sum_{P \in \mathsf{P}_{n}: |\Supp(P)| \leq k} \frac{\Tr(\rho P)}{2^n}P
    \end{align}
    Then for any $k \in [n]$,
    \begin{align}
        \|\rho-\rho_{sparse,k}\|_1 \leq e^{-k\log n}
    \end{align}
\end{theorem}

Next, we consider truncating terms of \Cref{eq:inclusion-exclusion} where $A$ contains a large set of sublattices that are all connected via a contiguous path of sublattices in $A$ (a `connected component' of sublattices). We call this a `percolated approximation.' When the depth exceeds a critical $\Theta(p^{-1}\log (p^{-1}))$ threshold, the decay of each inclusion-exclusion term of size $|A|$ will overwhelm the total number of connected components possible of size $|A|$, and so the approximation error of the percolated approximation can be bounded tightly. This is essentially due to a `site percolation' phase transition which was witnessed in circuits with non-universal gate sets \cite{rajakumar_watson_liu,nelson_rajakumar_naturally,oh_classical_simulability_linear_optical}, and arises from the fact that each sublattice has a bounded number of neighbors. Making this argument rigorous requires some careful grouping of terms in the inclusion-exclusion decomposition so that the norm does not blow up. In the end, we can prove the following,
\begin{theorem}[Convergence to percolated approximations] \label{theorem:convergence2}
    Let $\rho$ be the output state of a noisy geometrically local quantum circuit as in \Cref{def:circuit} with $d >d ^*$, where $d^* = \Theta(p^{-1}\log p^{-1})$, and let $J$ denote the coarse-graining of \Cref{def:coarse-grain}.. For any parameter $\ell \in [n]$, define $\rho_{perc,\ell} $ as follows,
    \begin{align}
    \rho_{perc,\ell} := \sum_{A \subseteq J:\text{ connected components in $A$ are $\leq \ell$}} \mc M_A(\rho_A) \otimes \sigma_{J \backslash A} \label{eq:percolated_approximation}    \end{align}
Or equivalently, in the Pauli basis, 
    \begin{align}  
    \rho_{perc,\ell} := \sum_{P \in \mathsf{P}_{n}: \text{ connected components in $\Supp(P)$ are $\leq \ell$}} \frac{\Tr(\rho P)}{2^n}P
\end{align}
where we define a `connnected component' as any set of sublattices in $A$ which are connected via a contiguous path of sublattices in $A$. Then,
    \begin{align}
        \|\rho-\rho_{perc,\ell}\|_1 \leq e^{-\Omega(pd\ell)}m
    \end{align}
\end{theorem}

\begin{remark} \label{remark:subset_perc}
    The above theorem can be modified so that for any $S \subseteq J$, truncation of all Pauli operators/inclusion-exclusion operators supported on $>\ell$-sized connected components entirely within $S$ will total trace distance error bounded by $e^{-\Omega(pd\ell)}|S|$. We outline this in the proof in \Cref{app:truncatability}. 
\end{remark}

\section{Classical Simulability at Logarithmic Critical Depth} 
\Cref{theorem:convergence1} tells us that when the depth exceeds a critical $\Theta(p^{-1}\log n)$ threshold, $\rho_{sparse, \Theta(\log n)}$ is inverse polynomially close to $\rho$ in trace distance. Note $\rho_{sparse, \Theta(\log n)}$ is not a true density matrix since it is not guaranteed to be positive semidefinite. However, by a standard sampling-to-computing reduction \cite{bremner_achieving}, if we can compute marginals of $\rho_{sparse, \Theta(\log n)}$, then we will be able to approximately sample from $\rho$. It turns out we can do this efficiently (in quasipolynomial time) using a classical algorithm that simply enumerates all the terms in $\rho_{sparse, \Theta(\log n)}$.

\begin{algorithm}[H]
\caption{Pauli-enumeration sampler}
\label{alg:pauli-enumeration}
\DontPrintSemicolon

\KwIn{
The classical description of a noisy geometrically local circuit and product input state that results in a quantum state $\rho$; a cutoff parameter $k$.
}
\KwOut{
A sample $b \in \{0,1\}^n$ from a distribution that is close to the quasiprobability distribution produced by measuring $\rho_{sparse,k}$, as defined in \Cref{theorem:convergence1}, in the computational basis.
}

\BlankLine
Let $Z_k = \{P \in \{I,Z\}^{\otimes n} : |\Supp(P)| \leq k\}$ and compute $\alpha_P \;=\; \Tr(P\,\rho_{\Supp(P)}) \forall P \in Z_k$ using brute force simulation of the reverse lightcone of $\rho_{\Supp(P)}$.

\For{$i=1$ \KwTo $n$}{
    Compute
    \[
    q(b_i= 1 \mid b_{<i}) = \frac{ \Tr(\rho_{sparse,k} [\ketbra{b_{<i}1_i} \otimes I_{>i}])}{\Tr(\rho_{sparse,k} [\ketbra{b_{< i}}\otimes I_{\geq i}])} = \frac{\sum_{P \in Z_k} \alpha_P \Tr(P [\ketbra{b_{<i}1_i}   \otimes I_{>i}])}{\sum_{P \in Z_k} \alpha_P \Tr(P [\ketbra{b_{< i}}\otimes I_{\geq i}])}
    \]
    
    If $q(b_i= 1 \mid b_{<i})$ is negative or greater than 1, set $b_i=0$ or $1$ respectively. Otherwise set $b_i = 1$ with probability $q(b_i= 1 \mid b_{<i})$ and $b_i = 0$ otherwise.
}
\Return{$b=b_1\cdots b_n$}\;
\end{algorithm}

\begin{theorem} \label{theorem:sampling}
    Let $P$ be the output distribution of any geometrically local noisy quantum circuit as in \Cref{def:circuit} with $d > d^*$, where $d^* = \Theta(p^{-1}\log (n))$. Then Algorithm \ref{alg:pauli-enumeration} samples from a distribution $Q$ such that $\|P-Q\|_1 \leq \epsilon$, and has runtime $(\frac{1}{\epsilon})^{O(d^D) /\log(n)}$
\end{theorem}
\begin{proof}
    First, we prove that Algorithm \ref{alg:pauli-enumeration} is accurate. We note that this algorithm is exactly the standard sampling-to-computing reduction used in \cite{bremner_achieving,aharonov_polynomial_2023}, where the quasiprobability distribution resulting from $\rho_{sparse,k}$ is `fixed' to be a true probability distribution, incurring an error proportional to $\|\rho - \rho_{sparse,k}\|_1$. By \Cref{theorem:convergence1}, it suffices to set $k = \log(\frac{1}{\epsilon})/\log(n)$ to obtain the desired approximation error. 
    
    Next, we prove that Algorithm \ref{alg:pauli-enumeration} is efficient. It is clear that there are $\sum_{k' \leq k} {m \choose k'}\leq k(2en/k)^{k}$ possible values of $\Supp(P)$ with size at most $k$. Since there are at most $k(4d)^D$ qubits in the reverse lightcone each $\rho_{\Supp(P)}$, we can brute force compute and enumerate the coefficients of every Pauli operator in any one of these terms in time $e^{O(kd^D)}$.  We need to compute $O(n)$ marginals in the bit-by-bit sampling algorithm of \cite{bremner_achieving}. Therefore, the runtime of this algorithm is
    \begin{align}
        O(nk(2en/k)^{k})e^{O(kd^D)}\leq  (\frac{1}{\epsilon})^{O(d^D) /\log(n)}
    \end{align}
\end{proof}

\section{Toward Constant-Depth Simulability and a Structural Phase Transition} \label{sec:conjecture}

The previous sections established a quasipolynomial-time simulation algorithm once the depth exceeds $\Theta(p^{-1}\log n)$, and have also proved a stronger structural statement at depth $\tilde{\Theta}(p^{-1})$: after coarse-graining the lattice into sublattices of size $\poly(d)$, Pauli weight supported on long connected structures can be truncated with exponentially small error. This raises a natural question: does this structural transition already imply classical simulability at constant depth?

In this section we propose a candidate route to such a result. We first describe a patching-type classical sampling algorithm and show that it succeeds whenever the output distribution satisfies an approximate Markov condition on the coarse-grained lattice. We then formulate a conjecture asserting that this condition holds for all noisy geometrically local circuits once $d>\tilde{\Theta}(p^{-1})$. Next, we explain why the structural phase transition proved earlier provides concrete evidence for this conjecture. Finally, we discuss consequences for measurement-free fault-tolerance constructions, showing that our structural results sharply constrain the kinds of long-range order that can persist above this depth scale.

\subsection{A Candidate Sampling Algorithm}

Our candidate algorithm is motivated by the following principle. If the output distribution of a noisy geometrically local circuit has sufficiently weak long-range conditional correlations, then one should be able to sample the output bitstring sequentially while conditioning only on a local neighborhood of the previously sampled qubits. In other words, once an appropriate boundary has been conditioned on, distant parts of the past should become approximately irrelevant.

This suggests a \emph{patching} strategy: partition the lattice into coarse-grained sublattices, order these sublattices, and sample them one at a time while truncating the conditioning region to an $\ell$-neighborhood in the coarse-grained adjacency graph. The key point is that if the required neighborhood radius $\ell$ is only logarithmic in $n$, then the relevant marginals can still be computed by brute-force simulation of a local reverse lightcone.

\begin{algorithm}[H]
\caption{Patching algorithm}
\label{alg:patching}
\DontPrintSemicolon

\KwIn{
The classical description of a noisy geometrically local circuit and product input state that results in a quantum state $\rho$ (and resulting output distribution $P$); a cutoff parameter $\ell$.
}
\KwOut{
A sample $b \in \{0,1\}^n$
}

\BlankLine

\For{$i=1$ \KwTo $m$}{
    For each $x \in \{0,1\}^{|J_i|}$, compute $P(x_{J_i}, b_{J_{< i} \cap \partial^\ell J_i})$ using brute force simulation of the reverse lightcone of $\rho_{J_{\leq i} \cap \partial^\ell J_i}$.
    
    Set $b_{J_i} = x$ with probability
    \[
    P(b_{J_i}= x \mid b_{J_{< i} \cap \partial^\ell J_i}) =  \frac{P(x_{J_i}, b_{J_{< i} \cap \partial^\ell J_i})}{\sum_{x \in \{0,1\}^{|J_i}}P(x_{J_i}, b_{J_{< i} \cap \partial^\ell J_i})}
    \]
}
\Return{$b=b_{J_i}\cdots b_{J_m}$}\;
\end{algorithm}

\begin{theorem} [Classical Simulability from Approximate Markov Condition]\label{theorem:simulability-from-markov}
    Let $P$ be the output distribution of any geometrically local quantum circuit as in \Cref{def:circuit}, and let $J$ denote the coarse-graining of \Cref{def:coarse-grain}. Suppose for every $i \in J$,
    \begin{align}
        I_P(J_i:J_{\leq i} \backslash \partial^\ell J_i|J_{\leq i} \cap \partial^\ell J_i) \leq e^{-\Omega(\ell)}|J_i|
    \end{align}
    Then Algorithm \ref{alg:patching} samples from a distribution $Q$ such that $\|P-Q\|_1 \leq \epsilon$, and has runtime $ne^{O((d\log (n/\epsilon))^D)})$.
\end{theorem}
\begin{proof}
First, we prove that Algorithm \ref{alg:patching} is accurate. We introduce some shorthand: for arbitrary probability distributions $P^{(1)},P^{(2)}$ on $n$ bits and subsets $A,B,C \subseteq [n]$, we use $P^{(1)}_{AB}P^{(2)}_{C|B}$ to denote the distribution where $\forall b \in \{0,1\}^n$, $P^{(1)}_{AB}P^{(2)}_{C|B}(b_{ABC}) := P_{AB}^{(1)}(b_{AB})P^{(2)}_{C|B = b_B}(b_{C}) $. Now, let $Q$ be the final distribution produced by Algorithm \ref{alg:patching}. $Q_{J_{\leq i}}$ is the distribution after step $i$.  We use an inductive argument. Clearly $P_{J_1} = Q_{J_1}$. At step $i$, we have $Q_{J_{ \leq i}} = Q_{J_{ < i}} P_{J_i | \partial^\ell J_{< i}}$, and we can simplify as follows,
    \begin{align}
        \|P_{J_{\leq i}} - Q_{J_{ \leq i}}\|_1 &= \|P_{J_{\leq i}} - Q_{J_{ < i}} P_{J_i | \partial^\ell J_{< i}}\|_1\\
        &= \|P_{J_{\leq i}} - P_{J_{ < i}} P_{J_i | \partial^\ell J_{< i}} + P_{J_{ < i}} P_{J_i | \partial^\ell J_{< i}} - Q_{J_{ < i}} P_{J_i | \partial^\ell J_{< i}}\|_1\\
        &\leq \|P_{J_{\leq i}} - P_{J_{ < i}} P_{J_i | \partial^\ell J_{< i}}\|_1 + \|P_{J_{<i}} P_{J_i | \partial^\ell J_{< i}} - Q_{J_{< i}} P_{J_i | \partial^\ell J_{< i}}\|_1\\
        &\leq \sqrt{2  I_P(J_i:J_{\leq i} \backslash \partial^\ell J_i|J_{\leq i} \cap \partial^\ell J_i)}  + \|P_{J_{<i}} P_{J_i | \partial^\ell J_{< i}} - Q_{J_{< i}} P_{J_i | \partial^\ell J_{< i}}\|_1 \tag{Pinsker's inequality}\\
        &\leq e^{-\Omega(\ell)}|J_i| + \|(P_{J_{<i}}  - Q_{J_{< i}}) P_{J_i | \partial^\ell J_{< i}}\|_1 \\
        &= e^{-\Omega(\ell)}|J_i| + \|P_{J_{<i}} - Q_{J_{< i}} \|_1
    \end{align}
    Clearly, the overall error is $m|A| e^{-\Omega(\ell)} = ne^{-\Omega(\ell)}$. Therefore, an approximation error of $\epsilon$ can be obtained by setting $\ell = \Theta(\log  (n/\epsilon)) $. 
    
    Next, we prove that Algorithm \ref{alg:patching} is efficient. Since at each step $i$, $\rho_{J_{\leq i} \cap \partial^\ell J_i}$ spans at most $O(\ell^D)$ sublattices, brute force simulation of its reverse lightcone requires runtime $e^{O(\ell^D d^D)}$. Substituting the value of $\ell$ in, the overall runtime is $ne^{O((d\log (n/\epsilon))^D)})$.
\end{proof}

\subsection{Conjectured Approximate Markov Structure at Constant Depth}
The theorem above suggests that the right object to conjecture is not merely the existence of some unknown simulation algorithm, but rather the validity of an approximate Markov property for the output distribution of every noisy geometrically local circuit once the depth exceeds the constant threshold suggested by our structural results.

\begin{conjecture} [Approximate Markov Condition at Constant Depth]\label{conj:markov}
    Let $P$ be the output distribution of any noisy geometrically local quantum circuit as in \Cref{def:circuit}. If $d > d^*$, where $d^* = \tilde{\Theta}(p^{-1})$, then the approximate markov condition of \Cref{theorem:simulability-from-markov} holds true. Consequently, Algorithm \ref{alg:patching} approximates $P$ in quasipolynomial time.
\end{conjecture}

We stress that this is a conjectural extension of our main theorem, not a proved consequence. However, it is strongly motivated by both the structural phase transition established in this work and by prior simulation results in non-universal models \cite{nelson_rajakumar_naturally}. In particular, if $P$ is the output distribution of any geometrically local quantum circuit consisting of Clifford gates, with input states and measurement in an arbitrary product basis, or IQP+CNOT gates, with input states and measurement in the hadamard basis, then the approximate markov condition of \Cref{theorem:simulability-from-markov} holds\footnote{Ref. \cite{nelson_rajakumar_naturally} shows that, with probability $1-e^{-\Omega(\ell)}$ over a \textit{classical} random variable $b$, $I_{P}(A:C|B,b) = 0$ ($b$ denotes error locations in the circuit). Then, $I_P(A:C|B) = \mathbb{E}_bI_{P}(A:C|B,b) \leq e^{-\Omega(\ell)}|A|$.} implying Algorithm \ref{alg:patching} accurately and efficiently approximates $P$.

\subsection{Evidence from Structural Phase Transition}

In the previous sections we proved that when $
d > \tilde{\Theta}(p^{-1})$,
the output state admits a coarse-grained Pauli truncation theorem: all Pauli operators supported on long connected structures can be discarded with exponentially small total error. This is exactly the depth scale at which one would expect a percolation-like transition to occur. We also remark that numerical evidence is fundamentally difficult to obtain because we would like results to apply to worse-case circuits -- results obtained from randomly quantum circuits would not suffice \footnote{moreover, approximate markovianity appears in noisy random quantum circuits for fundamentally different depth-independent reasons\cite{lee_ghosh_approx_markov,zhang_lee_approx_markov,wei_nelwon_rajakumar_2d_clifford}}.

We have reduced a conjecture on classical simulability to a conjecture about a concrete mathematical property of the output distribution. The remaining gap is to upgrade the proven decay of large-support Pauli weight into decay of conditional mutual information.

A plausible route forward is to consider three regions $A,B,C$ arranged so that $B$ forms a separator between $A$ and $C$ in the coarse-grained lattice. Consider a percolated approximation by truncating all Pauli operators that that form a path between $A$ and $C$, i.e. only keeping inclusion-exclusion terms with a disentangling boundary of depolarized sublattices between $A$ and $C$. This would incur an approximation error $\sim e^{-\ell}$ for any $\ell$ (see \Cref{remark:subset_perc}), and one might hope that the disentangling boundaries make it easier to prove an approximate markov condition.

A second, but more involved route forward is cluster expansion, which has been recently used to prove approximate markovianity in a variety of other settings \cite{bakshi_high-temperature,zhang_lee_approx_markov}. These techniques generally require exponential decay of high-weight connected `clusters', suitably defined, in a manner that overwhelms the number of possible connected clusters -- this is exactly what we have proven in \Cref{lemma:inclusion-exclusion} and \Cref{theorem:convergence2}. A key obstacle is that many cluster expansion ideas require tight bounds on operator norms, whereas our bounds deal with $\ell_1$ norms. Bounding operator norm by $\ell_1$ is quite lossy and does not seem to be tight enough when applied directly, e.g. to \cite{zhang_lee_approx_markov}.

\subsection{Consequences for Measurement-Free Fault Tolerance}

Our information-theoretic arguments rule out several candidate fault-tolerance constructions in geometrically local noisy circuits.  

A basic motivating example comes from Ref.~\cite{aharonov_limitations}, which shows that, in an all-to-all noisy quantum circuit with noise rate $p$, one can distill a fresh qubit in the state $\ket{0}$ at depth $d$ using exactly $3^d$ input qubits, with error probability $O(p)$. Informally, this allows the circuit to suppress noise accumulation up to depth $O(\log n)$ by repeatedly distilling fresh ancillas. \Cref{lemma:exponential_decay} implies that such a protocol cannot exist in geometrically local circuits. Indeed, once the depth exceeds the critical scale $\tilde{\Theta}(p^{-1})$, the output state $\rho$ is $e^{-pd}$-close to a product state of the form $\sigma_i \otimes \rho_{\bar{i}}$ for any qubit $i$. In particular, no single qubit can remain close to $\ket{0}$ for a long time.

This raises the question of whether more sophisticated fault-tolerance primitives can avoid this obstruction, even without intermediate measurements and resets. A prominent example is \emph{single-shot state preparation}, which we formalize below.

\begin{definition} [Single-Shot State Preparation]
    We will use $\mc Rec$ to denote a channel that performs measurement in an arbitrary basis on an arbitrary subset of the qubits (i.e. syndrome extraction) and applies an arbitrary recovery channel on the remaining qubits, labelled by $S_{\mathrm{prep}}$. We use $\mc Dec$ be a map (i.e. decoding procedure) from the physical space of $S_{\mathrm{prep}}$ to the logical space of some fixed code. 
\end{definition}

Although \Cref{def:circuit} forbids intermediate measurements and adaptive operations such as $\mc Rec$, single-shot state preparation is still relevant here because, in many settings, $\mc Rec$ and $\mc Dec$ can be deferred to classical post-processing by tracking the propagation of errors; see, for example, \cite{bravyi_gosset_koenig_tomamichel}. Our results imply a no-go theorem for such constructions above the critical depth threshold.

\begin{theorem}
    Let $\rho$ be the output state of a noisy geometrically local quantum circuit as in \Cref{def:circuit}, and $S_{\mathrm{prep}}$ be any subset of qubits. Consider stabilizer code with distance greater than some $\mathrm{dist}^* = \poly\log (|S_{\mathrm{prep}}|))$. Let $\ketbra{\bar{0}}$ and $\ketbra{\bar{1}}$ be orthogonal code states on $S_{\mathrm{prep}}$ that are distinguished by logical Pauli operator $\bar{Z}$ whose physical Pauli representatives are connected in the lattice (e.g. string-like, membrane-like). For any single-shot state preparation procedure involving $\mc Rec, \mc Dec$ and any constant $\epsilon$ reflecting the target fidelity, there exists some critical depth threshold $d^* = \tilde{\Theta}(p^{-1})$ (which crucially does not depend on $n$ or $|S_{\mathrm{prep}}|$), such that when $d > d^*$,
    \begin{align}
        \Tr(\mc Dec \circ \mc Rec (\rho) \bar{Z})< \epsilon
    \end{align}
    
\end{theorem}

\begin{proof}
    Because recovery and decoding are trace-distance contractive, \Cref{lemma:exponential_decay} together with Pinsker's inequality gives $
        \|\rho - \rho_{\overline{S_{\mathrm{prep}}}} \otimes \sigma_{S_{\mathrm{prep}}}\|_1 < \sqrt{\frac{1}{2}(1-p)^d |L(S_{\mathrm{prep}})|}$.
    Since measurements cannot reliably distinguish states that are close in trace distance, this already yields the theorem whenever $|S_{\mathrm{prep}}|=e^{O(d)}$ (for constant $p$). This means for the remainder of the proof, we can assume $d = O(\log(|S_{\mathrm{prep}}|))$. 
    
    Now consider a physical representative of the logical operator $\bar{Z}$. Every such representative must have support on at least $\mathrm{dist}$ qubits. Since the representative is connected, it must have support on a connected component of at least $m$ sublattices, where $m \geq \mathrm{dist}/(2d)^{D} = \mathrm{dist}/O(\log(|S_{\mathrm{prep}}|)^{D})$. Then, truncating of connected components of size $\geq m$ within $S_{\mathrm{prep}}$ incurs an error at most $e^{-\Omega(pdm)}|S_{\mathrm{prep}}|$ (see \Cref{remark:subset_perc}). Thus, there exists some $\mathrm{dist}^* = \poly\log(|S_{\mathrm{prep}}|)$, above which single-shot state prep is not possible.
\end{proof}

At the same time, this picture is close to tight. Below the critical depth threshold, there are constructions that satisfy many of the properties that we prove or conjecture to be impossible at larger depth. In particular, consider the single-shot state prep protocol of \cite{raussendorf_bravyi_harrington_cluster_states,bergamaschi_liu_single_shot, bravyi_gosset_koenig_tomamichel} that constructs a logical Bell state between encoded qubits on opposite surfaces of a 3D lattice, such that measuring all qubits in the bulk of the lattice produces enough syndrome information to correctly identify which physical Pauli errors occured on the opposite surfaces. Let $\rho$ be the output state of this noisy geometrically local constant-depth circuit, followed by $d$ layers of noisy identity gates. The total noise accumulated through the identity channels will remain below the fault tolerance threshold when $d < d^{**}$ for some $d^{**} = \Theta(p^{-1})$, which means (1) single-shot state preparation remains succesful whp, (2) the output distribution does not obey an approximate markov condition (since there are conditional correlations between opposite surfaces), and (3) the output state does not converge to its percolated approximation (since the logical observables provide strong signal, despite being string-like and stretching across many qubits). Thus, the contrast between the subcritical and supercritical regimes at constant depth reflects a phase transition in the types of nonlocal structure that the circuit can support. 

\section*{Acknowledgements}
This material is based upon work supported by the U.S. Department of Energy, Office of Science, Accelerated Research in Quantum Computing, Fundamental Algorithmic Research toward Quantum Utility (FAR-Qu). Additional support is acknowledged from IBM, where JR perfomed a portion of this research as an intern under the guidance of Abhinav Deshpande, Kunal Sharma, and Oles Shtanko. We thank Dominik Hangleiter, Zhi-Yuan Wei, Daniel Malz and Alexey Gorshkov for helpful discussions. JR thanks Yi-Kai Liu, Alexander Muller-Hermes, Thiago Bergamaschi, Su-un Lee, Frank Zhang, and Mohammad A. Alhejji for insight into the difficulties of proving approximate markovianity and encouraging us to present this work even without resolution of the main conjecture.
This material is based upon work supported in part by the NSF QLCI award OMA2120757. This work was performed in
part at the Kavli Institute for Theoretical Physics (KITP), which is supported by grant NSF PHY-2309135. JN is supported by the National Science Foundation Graduate Research Fellowship Program under Grant No. DGE 2236417.

\appendix
\section{Appendix}

\subsection{Information-Theoretic Arguments}

\subsubsection{Relative Entropy Convergence on Subsets} \label{app:relative_entropy_convergence_subsets}

The proof follows the basic outline of \cite{muller-hermes_stilck-franca_wolf_relative_entropy_convergence}, where an entropy production inequality is shown for tensor products of depolarizing channels. The main difference is that we will prove conditional entropy production, by applying the `conditional quantum shearer inequality' of \cite{berta_conditional_shearers} in the step where we would have applied the normal `quantum shearer inequality.'

\begin{fact} \label{fact:conditional}
For any state $\rho$ on $n$ qubits and $A \subseteq [n]$,
    \begin{align}
        D(\rho \|\sigma_A \otimes \rho_{\bar{A}}) = |A| - S(A|\bar{A})_\rho
    \end{align}
\end{fact}
\begin{proof}
    \begin{align}
         D(\rho \|\sigma_A \otimes \rho_{\bar{A}}) &=
         \Tr \rho \log \rho - \Tr \rho \log (\sigma_A \otimes \rho_{\bar{A}}) \\
         &= -S(\rho) - \Tr \rho(\eye \otimes \log \frac{\rho_{\bar{A}}}{2^{|A|}}) & \text{Since $\sigma_A = \eye/2^{|A|}$} \\
         &= -S(\rho) - \Tr \rho(\eye \otimes \log\rho_{\bar{A}}) + |A| \\
         &= |A|-S(\rho) + S(\rho_{\bar{A}}) \\
         &= |A| - S(A|\bar{A})_\rho
    \end{align}
\end{proof}
\begin{fact}[Conditional quantum Shearer inequality \cite{berta_conditional_shearers}] \label{fact:shearer}
    Consider $t \in \mathbb N$ and a family $\mc F$ of subsets of $[m]$ such that each $i \in [m]$ is contained in exactly $t$ elements of $\mc F$. Then for any $\rho \in \mathrm D(A_1 \dots A_m B)$ we have
    \[
        S(A_1 \dots A_m | B) \leq \frac{1}{t} \sum_{F \in \mc F} S(\{A_s\}_{s \in \mc F}|B)
    \]
\end{fact}

\begin{lemma} [Restatement of \Cref{lemma:relative_entropy_convergence_subsets}]
For any state $\rho$ on $n$ qubits and $A \subseteq [n]$,
\begin{align}
        D(\mc N_A(\rho)\|\sigma_A \otimes \rho_{\bar{A}}) \leq (1-p)D(\rho\|\sigma_A \otimes \rho_{\bar{A}}) \label{eq:relative}
    \end{align}
\end{lemma}
\begin{proof}
    Using Fact \ref{fact:conditional}, we can rewrite the given inequality as a conditional entropy production inequality. Note that the following is equivalent to Eq. \ref{eq:relative}.
    \begin{align}
    \label{eq:conditionalentropy}
        S(A|\bar{A})_{\mc N_A(\rho)} \geq (1-p) S(A|\bar{A})_{\rho} + p|A|
    \end{align}
Our goal will thus be to prove Eq. \ref{eq:conditionalentropy}. Let $\mc F_k := \{F\subseteq A : |F|=k\}$. Using the concavity of conditional entropy we can begin to bound the left-hand side as follows:
\begin{align}
    S(A|\bar{A})_{\mc N_A(\rho)} &\geq \sum_{k=0}^{|A|} p^k (1-p)^{|A|-k} \sum_{F \in \mc F_k} S(A|\bar{A})_{\sigma_F \otimes \rho_{\bar{F}}} \\
    &= \sum_{k=0}^{|A|} p^k (1-p)^{|A|-k}\sum_{F \in \mc F_k} (|F| + S(A|\bar{A})_{\rho_{\bar{F}}})
\end{align}
First, notice that 
\[
    \sum_{k=0}^{|A|} p^k (1-p)^{|A|-k} \sum_{F \in \mc F_k} |F| = p|A|
\]
This can be seen by noting that the left-hand side is exactly the expected number of depolarizing errors that occur in $A$. Next, we can use the conditional quantum Shearer inequality \Cref{fact:shearer} to bound the remaining term:
\begin{align}
    \sum_{k=0}^{|A|} p^k (1-p)^{|A|-k} \sum_{F \in \mc F_k} S(A|\bar{A})_{\rho_{\bar{F}}} &=\sum_{k=0}^{|A|} p^{|A|-k} (1-p)^{k} \sum_{F \in \mc F_k} S(A|\bar{A})_{\rho_{F\cup \bar{A}}} \\
    & \geq \sum_{k=0}^{|A|} p^{|A|-k} (1-p)^{k} {|A| \choose k} \frac{k}{|A|}S(A|\bar{A})_{\rho} \label{eq:applyshearer} \\
    &=(1-p)S(A|\bar{A}) \label{eq:applybin}
\end{align}
Where in \Cref{eq:applyshearer}, we applied the conditional quantum Shearer inequality (\Cref{fact:shearer}) for $t={|A| \choose k}\frac{k}{|A|}$, which is the number of subsets of $\mc F_k$ that each $i \in A$ appears. In \Cref{eq:applybin}, we applied the known identity: $\sum_{k=0}^{n} p^{n-k} (1-p)^{k} {n \choose k} \frac{k}{n} = 1-p$.
\end{proof}

\subsubsection{Relative Entropy Decay in Circuits} \label{app:exponential_decay}

Before stating our lemma, we define the notion of a reverse lightcone below (we adopt a quite formal definition which is useful in our proofs),

\begin{definition} [Reverse Lightcone] \label{def:lightcone}
    For any quantum circuit, for any set of qubits $A$, we use $L_{i}(A)$ to denote the qubits in the reverse lightcone of $A$ when considering only the last $d-i$ layers. The lightcone is defined inductively starting from the last layer of the circuit and moving backwards: $L_{d}(A) = A$ and $L_{i-1}(A)$ is the smallest superset of $L_i(A)$ such that no gate in $U^{(i)}$ crosses between $L_{i-1}$ and $\overline{L_{i-1}(A)}$. We will use $L(A)$ as shorthand for $L_0(A)$.
\end{definition}

First, we introduce a lemma that will allow us to consider supersets of $A$, e.g. $L_{d-1}(A)$, as an upper bound on the relative entropy between $\rho$ and $\sigma_A \otimes \rho_{\overline{A}}$. This is the mechanism by which we introduce dependence on lightcone size in our bounds.

\begin{fact}
\label{fact:switchsigma}
    Let $A \subseteq B$
    \begin{align}
        D(\rho || \sigma_{A} \otimes \rho_{\bar{A}}) \leq D(\rho || \sigma_B \otimes \rho_{\bar{B}})
    \end{align}
\end{fact}
\begin{proof}
    \begin{align}
        D(\rho || \sigma_{A} \otimes \rho_{\bar{A}}) -  D(\rho || \sigma_B \otimes \rho_{\bar{B}}) &= |A| - S(A|\bar{A})_{\rho} - |B| + S(B|\bar{B}) &\text{By \Cref{fact:conditional}} \\
        &= |A|-|B|+ S(\bar{A})_{\rho} - S(\bar{B})_{\rho} \\
        &= |A|-|B|+S(\bar{A} \backslash \bar{B} | \bar{B}) \\
        &\leq |A|-|B| + |\bar{A} \backslash \bar{B}|\\
        &= 0
    \end{align}
\end{proof}

Now, we prove the main lemma,
\begin{lemma} [Restatement of \Cref{lemma:exponential_decay}]
    Let $\rho$ be the output state of a noisy quantum circuit. Let $A \subseteq [n]$, and $L(A)$ be the set of qubits in its reverse lightcone.
    \begin{align}
        D(\rho\|\sigma_A \otimes \rho_{\overline{A}}) \leq (1-p)^d|L(A)|
    \end{align}
\end{lemma}
\begin{proof}
    Our proof strategy is to use each layer of noise to decay the relative entropy by $1-p$ and then argue that the intermediary layers of gates do not increase the relative entropy. We will start from the end of the circuit and `peel' back layers. Let $\rho^{(i)}$ denote the state after only the first $i$ layers of gates and subsequent noise channels are applied. Next, we introduce some notation to isolate only the noise channels within the lightcone of $A$. Let $\rho^{*(i)}$ denote the state after only the first $i$ layers of gates and subsequent noise channels are applied, \textit{except the last noise channels that act on $L_i(A)$}, i.e.
    \begin{align}
         \rho^{(i)} &= \mc N_{L_i(A)}^{(i)}(\rho^{*(i)}) \label{eq:star1}\\
         \rho^{*(i)} &= \mc N_{\overline{L_{i}(A)}}^{(i)} \circ \mc U^{(i)} (\rho^{(i-1)}) \label{eq:star2}
    \end{align}
    An important observation is that in our definition of the lightcone \Cref{def:lightcone}, $U^{(i)}$ cannot cross between $L_{i-1}(A)$ and $\overline{L_{i-1}(A)}$. This gives us the crucial property that all channels in layer $i$ commute with the completely depolarizing channel on $L_{i-1}(A)$. In particular,
    \begin{align}
        \mc D_{L_{i-1}(A)}  \circ \mc N_{\overline{L_{i}(A)}} \circ \mc U^{(i)} = \mc N_{\overline{L_{i}(A)}} \circ \mc U^{(i)} \circ \mc D_{L_{i-1}(A)} \label{eq:commutation}
    \end{align}

    The main inductive step is as follows, 
    \begin{align*}
        &D(\rho^{(i)} || \sigma_{L_i(A)} \otimes \rho_{\overline{L_i(A)}}^{(i)}) \\
        &= D(\mc N_{L_i(A)}(\rho^{*(i)}) || \sigma_{L_i(A)} \otimes \rho^{(i)}_{\overline{L_{i}(A)}} ) \\
        &= D(\mc N_{L_i(A)}(\rho^{*(i)}) || \sigma_{L_i(A)} \otimes \rho^{*(i)}_{\overline{L_{i}(A)}} ) 
        &\tag{By \Cref{eq:star1}}\\
        &\leq (1-p) D(\rho^{*(i)}||\sigma_{L_i(A)} \otimes \rho^{*(i)}_{\overline{L_{i}(A)}})
        &\tag{By \Cref{lemma:relative_entropy_convergence_subsets}}\\
        &\leq (1-p) D(\rho^{*(i)}||\sigma_{L_{i-1}(A)} \otimes \rho^{*(i)}_{\overline{L_{i-1}(A)}})
        &\tag{By \Cref{fact:switchsigma}}\\
        &= (1-p)D(\rho^{*(i)}||\mc D_{L_{i-1}(A)} (\rho^{*(i)}))\\
        &= (1-p) D(\mc N_{\overline{L_{i}(A)}}^{(i)} \circ \mc U^{(i)} (\rho^{(i-1)})||\mc D_{L_{i-1}(A)}  \circ \mc N_{\overline{L_{i}(A)}} \circ \mc U^{(i)}(\rho^{(i-1)}) )
        &\tag{By \Cref{eq:star2}}\\
        &= (1-p) D(\mc N_{\overline{L_{i}(A)}}^{(i)} \circ \mc U^{(i)} (\rho^{(i-1)})|| \mc N_{\overline{L_{i}(A)}} \circ \mc U^{(i)} \circ \mc D_{L_{i-1}(A)} (\rho^{(i-1)}) )
        &\tag{By \Cref{eq:commutation}}\\
        &= (1-p) D( \rho^{(i-1)}||\mc D_{L_{i-1}(A)} (\rho^{(i-1)}) )
        &\tag{monotonicity of relative entropy}\\
        &= (1-p) D(\rho^{(i-1)}\|\sigma_{L_{i-1}(A)} \otimes \rho^{(i-1)}_{\overline{L_{i-1}(A)}})
    \end{align*}
    Starting from $i = d$ and going to $i=0$, we get
    \begin{align*}
        D(\rho\|\sigma_A \otimes \rho_{\overline{A}}) &\leq (1-p)^d D(\rho^{(0)}\|\sigma_{L_{0}(A)} \otimes \rho^{(0)}_{\overline{L_{0}(A)}}) \\
        &= (1-p)^d D(\ketbra{0}_J \|\sigma_{L(A)} \otimes \ketbra{0}_{\overline{L_{0}(A)}}) \\
        &= (1-p)^d |L(A)|
    \end{align*}

\end{proof}

\subsection{Truncation Arguments} \label{app:truncatability}
\begin{lemma} [Restatement of \Cref{lemma:inclusion-exclusion}]
        Let $\rho$ be the output state of a noisy geometrically local quantum circuit.
        \begin{align}
            \|\mc M_A(\rho)\|_1 &\leq (2[(1-p)^d(4d)^D]^{1/(2 * 3^D)})^{|A|} 
        \end{align} 
\end{lemma}
\begin{proof}
    First, note that $\|\mc M_A(\rho)\|_1 \leq 2^{|A|}$, since there are $2^{|A|}$ terms in $\mc M_A(\rho)$, and they each have trace norm 1 (triangle inequality). Therefore, the bound is trivially true when $(1-p)^d(4d)^D \geq 1$. We will now consider the regime where $(1-p)^d(4d)^D < 1$. Let $A_{span}$ be the largest subset of $A$ such that no two sublattices in $A_{span}$ are adjacent. Note that $|A_{span}| \geq |A|/3^{D}$. Let $\Phi_{A_{span}}$ denote the channels of $\Phi$ in the reverse lightcone of $A_{span}$ and $\Phi'_{A_{span}}$ denote all channels applied subsequently in $\Phi$. Notice that because the sublattices are $2d$ qubits wide, the reverse lightcones of each sublattice in $A_{span}$ do not intersect, and therefore $\Phi_{A_{span}}$ has a product structure. In particular, if we define $\Phi_i$ as the channels of $\Phi$ in the reverse lightcone of any $i \in J$, we have,
    \begin{align}
        \Phi_{A_{span}}(\ketbra{0}_J) = \bigotimes_{i \in A_{span}} (\Phi_{i}(\ketbra{0}_{L(i)})) \otimes \ketbra{0}_{J \backslash L(A_{span})}\label{eq:product_structure}
    \end{align}
    The key idea is to use the contractivity of trace distance under CPTP maps for traceless operators \cite{perez-garcia_wolf_petz_ruskai_contractivity}. Note that the inclusion-exclusion sum is traceless, which allows us to use this inequality.
    \begin{align*}
        \|\mc M_A(\rho)\|_1  &= \|\sum_{B \subseteq 
        A}(-1)^{|B|} \rho_{J\backslash B} \otimes \sigma_{B}\|_1 \\
        &= \|\sum_{C \subseteq A\backslash A_{span}}(-1)^{|C|}\sum_{B\subseteq 
        A_{span}}(-1)^{|B|} \rho_{J\backslash BC} \otimes \sigma_{BC}\|_1\\
        &\leq  \sum_{C \subseteq A\backslash A_{span}} \|\sum_{B\subseteq 
        A_{span}}(-1)^{|B|} \rho_{J\backslash BC} \otimes  \sigma_{BC}\|_1 \tag{triangle inequality}\\
        &=  \sum_{C \subseteq A\backslash A_{span}} \|\sum_{B\subseteq 
        A_{span}}(-1)^{|B|} \mc D_{BC} (\rho)\|_1\\
        &=  \sum_{C \subseteq A\backslash A_{span}} \|\sum_{B\subseteq 
        A_{span}}(-1)^{|B|} \mc D_{BC} \circ \Phi'_{A_{span}} \circ \Phi_{A_{span}} (\ketbra{0}_J)\|_1\\
        &=  \sum_{C \subseteq A\backslash A_{span}} \| \mc D_{C} \circ \Phi'_{A_{span}}( \sum_{B\subseteq 
        A_{span}}(-1)^{|B|}  \mc D_{B}  \circ \Phi_{A_{span}} (\ketbra{0}_J))\|_1\\
        &\leq  \sum_{C \subseteq A\backslash A_{span}} \| \sum_{B\subseteq 
        A_{span}}(-1)^{|B|} \mc D_{B}  \circ \Phi_{A_{span}} (\ketbra{0}_J)\|_1 \tag{contractivity of trace distance}\\
        &= 2^{|A \backslash A_{span}|} \| \sum_{B\subseteq 
        A_{span}}(-1)^{|B|} [\bigotimes_{y \in B} \mc D_{y}]  \circ \Phi_{A_{span}} (\ketbra{0}_J))\|_1 \\
        &= 2^{|A \backslash A_{span}|} \|  \bigotimes_{y \in A_{span}} (\Phi_{y} (\ketbra{0}_{L(y)}) - \mc D_{y}  \circ \Phi_{y} (\ketbra{0}_{L(y)}))\|_1\tag{\Cref{eq:product_structure} and contractivity}\\
        &= 2^{|A \backslash A_{span}|} \prod_{y \in A_{span}} \|  \Phi_{y} (\ketbra{0}_{L(y)}) - \mc D_{y}  \circ \Phi_{y} (\ketbra{0}_{L(y)})\|_1 \tag{multiplicativity of trace norm under tensor product}\\
        &\leq 2^{|A \backslash A_{span}|} \prod_{y \in A_{span}} \sqrt{2D(\Phi_{y} (\ketbra{0}_{L(y)})||\mc D_{y}  \circ \Phi_{y} (\ketbra{0}_{L(y)}))} \tag{Pinsker's inequality}\\
        &\leq 2^{|A \backslash A_{span}|} (2(1-p)^d(4d)^D)^{|A_{span}|/2}\\
        &\leq (2[(1-p)^d(4d)^D]^{1/(2 * 3^D)})^{|A|} \tag{because $(1-p)^d(4d)^D < 1$}
    \end{align*}
\end{proof}

\begin{theorem}[Restatement of \Cref{theorem:convergence1}]
    Let $\rho$ be the output state of a noisy geometrically local quantum circuit with $d >d ^*$, where $d^* = \Theta(p^{-1}\log (n))$. Let $\rho_{sparse,k}$ be defined as in \Cref{eq:sparse_approximation}. Then,
    \begin{align}
        \|\rho-\rho_{sparse,k}\|_1 \leq e^{-k\log n}
    \end{align}
\end{theorem}
\begin{proof}
    We have,
    \begin{align*}
        \|\rho - \rho_{sparse,k}\| &= \|\sum_{A \subseteq J, |A| > k}\sigma_{J \backslash A} \otimes \mc M_A(\rho_A)\|_1\\
        &\leq \sum_{A \subseteq J, |A| > k} \|\sigma_{J \backslash A} \otimes \mc M_A(\rho_A)\|_1 \tag{triangle inequality}\\
        &\leq \sum_{A \subseteq J, |A| > k} (2[(1-p)^d(4d)^D]^{1/(2 * 3^D)})^{|A|} \\
        &\leq \sum_{k' > k} {m \choose k'}(2[(1-p)^d(4d)^D]^{1/(2 * 3^D)})^{k'} \\
        &\leq \sum_{k' > k} (\frac{em2[(1-p)^d(4d)^D]^{1/(2 * 3^D)}}{k'})^{k'} \\
        &\leq \sum_{k' > k} n^{-k'} \tag{\Cref{fact:critical_depth}, when $d > d^*$ and $d^* = \Theta(p^{-1}\log (n))$}\\
        &\leq n^{-k}
    \end{align*}
\end{proof}

\begin{theorem}[Restatement of \Cref{theorem:convergence2}]
    Let $\rho$ be the output state of a noisy geometrically local quantum circuit with $d >d ^*$, where $d^* = \Theta(p^{-1}\log p^{-1})$. Let $\rho_{perc,\ell}$ be defined as in \Cref{eq:percolated_approximation}. Then,
    \begin{align}
        \|\rho-\rho_{perc,\ell}\|_1 \leq e^{-\Omega(p\ell d)}m
    \end{align}
\end{theorem}
\begin{proof}
    Let us use $CC^{>\ell}(J_i)$ to denote the set of all possible $C \subseteq J$, such that $C$ is a connected component of more than $\ell$ sublattices that includes $J_i$ . We can now formalize the definition of $\rho_{perc,\ell}$ according to \Cref{eq:percolated_approximation},
    \begin{align}
        \rho_{perc,\ell} &= \sum_{A \subseteq J: \forall C \in \bigcup_{j \in [m]} CC^{> \ell}(J_j), C \not \subseteq A} \sigma_{J \backslash A} \otimes \mc M_A(\rho_A)
    \end{align}
    Now, we want to bound the error this truncation causes. For this, we will perform the truncation iteratively for each $i \in [m]$ (or in the case of \Cref{remark:subset_perc}, for each $i$ in some fixed subset of $J$). In particular, we will consider a sequence of matrices $\{\rho^{(i)}\}$, which are defined by truncating all terms with $A$ containing connected components from $\bigcup_{j \in [i]}CC^{> \ell}(J_j)$. This is formalized below
    \begin{align}
        \rho^{(i)} = \sum_{A \subseteq J:  \forall C  \in \bigcup_{j \in [i]}CC^{> \ell}(J_j), C \not \subseteq A} \sigma_{J \backslash A} \otimes \mc M_A(\rho_A) \label{eq:rho_i}
    \end{align}
    where we set $\rho^{(0)} = \rho$. We will consider $\rho_{perc,\ell} = \rho^{(m)}$. Now, by inspecting \Cref{eq:rho_i}, we can write,
    \begin{align}
        \rho^{(i)} - \rho^{(i-1)} &= \sum_{\substack{A \subseteq J: \exists C \in CC^{> \ell}(J_i) ,C  \subseteq A, \\
        \forall C' \in \bigcup_{j \in [i-1]} CC^{> \ell}(J_j), C' \not \subseteq A}} \sigma_{J \backslash A} \otimes \mc M_A(\rho_A)  \label{eq:diff}
    \end{align}
    Essentially, this sums over all terms in the inclusion-exclusion decomposition of $\rho$, where $A$ contains a large connected component around $J_i$, but does not contain large connected components around any $J_j$ for all $j \in [i-1]$, since these have already been truncated in steps $1\ldots i-1$. Note, there is a combinatorially large number of terms that are truncated at any step, and so it is not possible to directly bound the trace norm of \Cref{eq:diff} using the triangle inequality and \Cref{lemma:exponential_decay} to bound each term of the sum. To resolve this, we present a method to group together terms in the above sum, so that we can bound the error sufficiently. 
    
    In particular, at any truncation step $i$, we iterate over all possible $C \in CC^{> \ell}(J_i)$ and for each such $C$, we group together every valid choice of $A$ where $C \subseteq A$. To make this an accurate enumeration of terms in \Cref{eq:diff}, we further enforce that $A$ satisfies the following two conditions,
    \begin{itemize}
        \item $\forall C' \in \bigcup_{j \in [i-1]} CC^{> \ell}(J_j), C' \not \subseteq A$
        \item $C$ is the \textit{largest} connected component in $A$ that contains $J_i$
    \end{itemize}
    The latter condition ensures that different choices of $C$ will never correspond to the same choice of $A$, so this grouping technique does not perform any double-counting. To enforce the latter condition, it is enough to require that $\partial C \subseteq J \backslash A$, i.e. that the boundary of $C$ is not in $A$ and is therefore completely depolarized. Note that in general, $A$ then has the form $A = CA'$, where $A' \subseteq J \backslash (C \partial C)$. Our key strategy will be group together all valid choices of $A'$ (the region outside the depolarized boundary $\partial C$) for any choice of $C$ (the region inside the depolarized boundary $\partial C$). Explicitly,

    \begin{align*}
        &\rho^{(i)} - \rho^{(i-1)} \\
        &=  \sum_{\substack{C \in CC^{> \ell}(J_i) }} \left(\sum_{\substack{A \subseteq J: C \subseteq A,\\ \forall C' \in \bigcup_{i \in [i-1]} CC^{> \ell}(J_j), C' \not \subseteq A}}  \mc M_{A}(\rho_{A}) \otimes \sigma_{J \backslash (A)} \right)\\
        &=  \sum_{\substack{C \in CC^{> \ell}(J_i) \backslash \bigcup_{j \in [i-1]} CC^\ell(J_j) \\
        }} \left(\sum_{\substack{A' \subseteq J \backslash (C \partial C): \\ \forall C' \in \bigcup_{i \in [i-1]} CC^{> \ell}(J_j), C' \not \subseteq A'}}  \mc M_{A'  C}(\rho_{A' C}) \otimes \sigma_{J \backslash (A'C)}\right) \tag{because $A = CA'$}\\
        &= \sum_{\substack{C \in CC^{> \ell}(J_i) \backslash \bigcup_{j \in [i-1]} CC^\ell(J_j) \\
        }} \left(\sum_{\substack{A' \subseteq J \backslash (C \partial C): \\ \forall C' \in \bigcup_{i \in [i-1]} CC^{> \ell}(J_j), C' \not \subseteq A'}}  \mc M_{A' }(\rho_{A'}) \otimes \sigma_{J \backslash (A'C)} \otimes \mc M_C(\rho_C) \tag{because $\partial C \not \in A'C$}\right)\\
        &=  \sum_{\substack{C \in CC^{> \ell}(J_i) \backslash \bigcup_{j \in [i-1]} CC^\ell(J_j) \\
        }} \left(\sum_{\substack{A' \subseteq J \backslash (C \partial C): \\ \forall 'C \in \bigcup_{i \in [i-1]} CC^{> \ell}(J_j), C' \not \subseteq A'}}  \mc M_{A' }(\rho_{A'}) \otimes \sigma_{J \backslash (A'C\partial C)}\right) \otimes \sigma_{\partial C} \otimes \mc M_C(\rho_C)  \\
        &= \sum_{\substack{C \in CC^{> \ell}(J_i) \backslash \bigcup_{j \in [i-1]} CC^\ell(J_j) \\
        }} \rho^{(i-1)}_{J \backslash C \partial C} \otimes \sigma_{\partial C} \otimes \mc M_C(\rho_C) \tag{by \Cref{eq:rho_i}}\\
    \end{align*}
    This grouping allows us to sum over all possible connected components $C$, rather than all possible $C$ and all possible $A'$, as would be necessary if we directly attempted to bound \Cref{eq:diff}. The number of possible $C$ can be tightly bounded due to geometric locality, and it is overwhelmed by the decay of each term. Explicitly,
    \begin{align*}
        \norm{\rho^{(i)} - \rho^{(i-1)}}_1 &\leq \sum_{\substack{C \in CC^{> \ell}(J_i) \backslash \bigcup_{j \in [i-1]} CC^\ell(J_j) \\
        }} \norm{\rho^{(i-1)}_{J \backslash C \partial C}}_1 \norm{\mc M_C(\rho_C)}_1 \tag{triangle inequality and submultiplicativity}\\
        &\leq \sum_{x > \ell} \sum_{\substack{C \in CC^{x}(J_i) \backslash \bigcup_{j \in [i-1]} CC^\ell(J_j) \\
        }} \norm{\rho^{(i-1)}}_1 (2[(1-p)^d(4d)^D]^{1/(2 * 3^D)})^x \tag{contractivity of trace norm and \Cref{lemma:inclusion-exclusion}}\\
        &\leq \sum_{x > \ell} (3^D)^x \norm{\rho^{(i-1)}}_1 (2[(1-p)^d(4d)^D]^{1/(2 * 3^D)})^x  \tag{because each sublattice has $<3^D$ neighbors}\\
        &\leq e^{-\Omega(p\ell d)} \norm{\rho^{(i-1)}}_1
    \end{align*}
    Where in the final step, we have applied \Cref{fact:critical_depth}, assuming $d > d^*$, where $d^* = \Theta(p^{-1}\log p^{-1})$. We then have,
    \begin{align}
        \|\rho^{(i)} - \rho\|_1 &\leq \norm{\rho^{(i)} - \rho^{(i-1)}}_1 + \norm{\rho^{(i-1)} - \rho}_1 \tag{triangle inequality}\\
        &\leq e^{-\Omega(p\ell d)} \norm{\rho^{(i-1)}}_1 + \norm{\rho^{(i-1)} - \rho}_1\\
        &\leq e^{-\Omega(p\ell d)} (\norm{\rho}_1+\norm{\rho^{(i-1)} - \rho}_1) + \norm{\rho^{(i-1)} - \rho}_1 \tag{triangle inequality} \\
        &= e^{-\Omega(p\ell d)} + (1+e^{-\Omega(p\ell d)})\norm{\rho^{(i-1)} - \rho}_1
    \end{align}
    Solving this recursive formula with the base case of $\|\rho^{(0)} - \rho\| = 0$ and recursing until $i = m$ where $\rho^{(i)} = \rho_{perc,\ell}$ (or in the case of \Cref{remark:subset_perc}, only recurse until some fixed subset of $J$ is completed), we get,
    \begin{align}
        \|\rho_{perc,\ell} - \rho\|_1 &\leq (1+e^{-\Omega(p\ell d)})^m -1 \\
        &\leq e^{e^{-\Omega(p\ell d)}m} -1 \\
        &\leq 2e^{-\Omega(p\ell d)}m \tag{when $e^{-\Omega(p\ell d)} m\leq 1$}
    \end{align}
\end{proof}

\begin{lemma}[Restatement of \Cref{lemma:pauli_basis}]
For any $A \subseteq J$ and $P \in \mathsf{P}_n$,
\begin{align}
    \mc M_A (P) = \begin{cases}
            P, &\text{if }A \subseteq \Supp(P)\\
            0, &\text{if }A \not \subseteq \Supp(P)
        \end{cases}
\end{align}
Moreover,
\begin{align}
        \mc M_A \otimes \mc D_{J \backslash A} (P) = \begin{cases}
            P, &\text{if }A = \Supp(P)\\
            0, &\text{if }A \neq \Supp(P)
        \end{cases}
    \end{align}
\end{lemma}
\begin{proof}
    First, note that $\mc D_i (P)$ has the following effect,
    \begin{align}
        \mc D_i(P) &= \sigma_i \otimes \Tr_i(P)\\
        &= \begin{cases}
            P, &\text{if }i \not \in \Supp(P)\\
            0, &\text{if }i \in \Supp(P)
        \end{cases}
    \end{align}
    Next, we consider the action of $\mc M_i$,
    \begin{align}
        \mc M_i(P) &= [\mc I - \mc D_i](P)\\
        &= \begin{cases}
            P, &\text{if }i \in \Supp(P)\\
            0, &\text{if }i \not \in \Supp(P)
        \end{cases}
    \end{align}
    Thus, if we compose $\mc D_i$ foreach $i \in J \backslash A$ and $\mc M_i$ for each $i \in A$, we get the result described.
\end{proof}

\subsection{Observable Estimation in Noisy Geometrically Local Quantum Circuits} \label{app:observable}
First, we describe how Pauli observables and marginal probabilities experience exponential decay due to even a single layer of depolarizing noise, and then we show how this allows them to be efficiently estimated in noisy geometrically local quantum circuits.
\begin{lemma} \label{lemma:concentration}
    Let $\rho$ be the output state of any quantum circuit with a single layer of depolarizing noise on each qubit of strength $p$ as the final layer. Let $O$ be any observable on any subset of qubits $A \subseteq [n]$ which is either (1) a Pauli operator or (2) a projector onto a bitstring. Then, $|\Tr(\rho O)| \leq (1-p)^{|A|}$
\end{lemma}
\begin{proof}
    Let $\rho'$ be the state right before the final layer of depolarizing noise. We have,
    \begin{align}
        |\Tr(\rho O)| &= |\Tr(\mc N_{[n]}(\rho') O)|\\
        &= |\Tr(\rho' \mc N_{[n]}(O))| \tag{the adjoint map of a depolarizing channel is depolarizing}\\
        &\leq \|\rho'\|_1\|N_{[n]}(O)\|_\infty \tag{holder's inequality}\\
        &\leq (1-p)^{|A|} 
    \end{align}
    Where in the final step we have used the observation that the operator norm of Pauli observables and projectors onto bitstring decay in the size of their support under tensor products of depolarizing channels.
\end{proof}

\begin{corollary}
    Let $\rho$ be the output state of a noisy geometrically local quantum circuit and $O$ be any observable on any subset of qubits $A \subseteq [n]$ which is either (1) a Pauli operator or (2) a projector onto a bitstring. There exists a classical algorithm that outputs a number $\mu$, such that $|\mu - \Tr(\rho O)| \leq \epsilon$, in runtime $(\frac{1}{e})^{O( \min(d^D,(\log n)^D))}$.
\end{corollary}
\begin{proof}
   Recall, if the circuit is of depth $d = \omega(\log n)$, then the estimate $\Tr(O\sigma)$ is a good estimate due to the trace distance convergence of the output of noisy quantum circuits to the maximally mixed state. Therefore, we only need to handle the case that $d = O(\log n)$. It is clear from \Cref{lemma:concentration} that when $|A|$ is larger than some $O(\log (\frac{1}{\epsilon}))$ size, $\mu = 0$ is a sufficient estimate of the observable. Thus, we also only need to handle the case that $|A| = O(\log (\frac{1}{\epsilon}))$. Since there are at most $(2d)^D|A|$ qubits in the reverse lightcone of $A$, brute force simulation results in the runtime claimed.
\end{proof}

We next point out existing results of \cite{coble_coudron} and extensions \cite{dontha_tan_coudron_approximating}, which provides a classical method to estimate, within inverse polynomial error, any output probability of a \textit{noiseless} geometrically local quantum circuits in runtime $ \sim e^{O((d \poly\log (n))^{D \cdot 3^D})}$. It was noted in \cite{coble_coudron}, that this implies efficient algorithms to estimate the magnitude (but not the sign) of expectation values of Pauli operators and marginal probabilities as well. These are much stronger results since they do not require noise. They can be easily applied to our setting using `monte-carlo' sampling: for each depolarizing channel of strength $p$, sample an $X$, $Y$, or $Z$ \textit{unitary} error with probability $p/4$. In expectation, this simulates the true noisy circuit. Since estimates of Pauli's and marginals are always between $1$ and $-1$, a polynomial number of samples suffices to obtain an inverse polynomial additive error (by Hoeffding's inequality).  

\begingroup
		\printbibliography[heading=bibintoc]
\endgroup

\end{document}